\documentclass[aps,superscriptaddress,twocolumn,twoside,floatfix,pra,a4paper]{revtex4-2}

\usepackage[charter,cal=cmcal,sfscaled=false]{mathdesign}

\usepackage{dcolumn}
\usepackage{bm}
\usepackage{dsfont}
\usepackage{graphicx,epsfig}
\usepackage{amsmath}
\usepackage{mathtools}
\usepackage{braket}
\usepackage{physics}
\usepackage{caption}
\usepackage{float}
\usepackage{times}
\usepackage{soul}
\usepackage[shortlabels]{enumitem}

\usepackage{blindtext}
\usepackage{color}

\usepackage[colorlinks]{hyperref}
\hypersetup{
	colorlinks = true,
	urlcolor = {blue},
	citecolor = {blue},
	linkcolor= {blue}
}

\newcommand{\proj}[1]{\ket{#1}\bra{#1}}

\renewcommand{\eqref}[1]{Eq.~(\ref{#1})}

\newcommand{\bracket}[3]{\langle#1|#2|#3\rangle}

\newcommand{\mcE}{{\mathcal{E}}}
\newcommand{\mcN}{{\mathcal{N}}}
\newcommand{\AtoB}[2]{{#1\rightarrow#2}}

\newtheorem{theorem}{Theorem}

\newtheorem{lemma}[theorem]{Lemma}
\newtheorem{result}{Result}

\newcommand{\CY}[1]{{\color{black}#1}}
\newcommand{\CYtwo}[1]{{\color{black}#1}}
\newcommand{\CYSuppMat}[1]{{\color{black}#1}}

\begin{document}

\title{Resource theory of interactive quantum instruments}

\author{Chung-Yun Hsieh}
\email{chung-yun.hsieh@bristol.ac.uk}
\affiliation{H. H. Wills Physics Laboratory, University of Bristol, Tyndall Avenue, Bristol, BS8 1TL, United Kingdom}

\author{Armin Tavakoli}
\affiliation{Physics Department and NanoLund, Lund University, Box 118, 22100 Lund, Sweden}

\author{Huan-Yu Ku}
\affiliation{Department of Physics, National Taiwan Normal University, Taipei 11677, Taiwan}

\author{Paul Skrzypczyk}
\affiliation{H. H. Wills Physics Laboratory, University of Bristol, Tyndall Avenue, Bristol, BS8 1TL, United Kingdom}

\begin{abstract}
Quantum instruments describe both the classical outcome and the updated quantum state in a measurement process. To do this in a non-trivial way, instruments must have the capability to interact coherently with the state that they measure. Here, we develop a resource theory for instruments. We consider a relevant quantifier of the separation between interactive and non-interactive instruments and show that it admits three distinct operational interpretations in terms of quantum information tasks. These concern (i) the preservation of maximally entangled states after a local measurement, (ii) the average ability to preserve random states after measurement, and (iii) the ability to recover the classical information generated from measuring half of a maximally entangled state. We also introduce a natural set of allowed operations and show that the third task fully characterises the resource content of instruments. Our general framework reproduces as special cases established resource theories for channels and measurements.
\end{abstract}

\maketitle

\section{Introduction} 

Measurements, formally known as positive operator-valued measures (POVMs), are used to extract classical information from quantum systems.  In quantum information science, they are often viewed as resources \cite{ChitambarRMP2019}, that can have different capabilities depending on their specific features. For example, one may view as a resource measurements that cannot be simulated with projections~\cite{Khandelwal2025PRL}, sets of measurements that cannot be implemented jointly or without superposition~\cite{Otfried2021Rev,Hsieh-IP,Skrzypczyk2019,Armin2020,Buscemi2020PRL,Ji2024PRXQ,Takagi2019,Oszmaniec2019Quantum,Ducuara2020PRL,Uola2020PRL,Ducuara2022PRXQ,Ku2022NC,Uola2015PRL,Cavalcanti2016}, and measurements that require some type of coherent interaction with the state \cite{Skrzypczyk2019PRL}. 

However, measurement processes can be more complex than just producing classical information: fundamentally, they are also associated with a post-measurement state. Describing both the classical outcome and the quantum output in a measurement process requires a more general notion than POVMs called a quantum instrument. This richer and more complex structure has motivated many investigations into the capabilities of quantum instruments; see, e.g., Refs.~\cite{Hsieh2024PRL,Ji2024PRXQ,Ghai2025,Buscemi2023,Leppajarvi2024Quantum,Kuroiwa2024PRA}. 

The perhaps most basic feature needed to make a quantum instrument genuinely useful is that it is able to interact coherently with the state that it operates on. In other words,  the generation of the classical outcome and the associated quantum output in the measurement process cannot be separated from the incoming state. In view of this, it is natural to ask whether these non-trivial instruments, which we refer to as {\em interactive}, can be characterised, quantified and compared in operationally meaningful ways. While some progress has been made towards quantifying the interactive capabilities of instruments~\cite{Ghai2025}, a complete operational framework is missing.

Here, we develop a complete operational resource theory~\cite{ChitambarRMP2019} of interactive quantum instruments. We introduce a robustness measure to quantify the interactivity of instruments and then prove that this abstract quantity admits three different operational interpretations. Specifically, we show that it benchmarks the performance of instruments in three different quantum information tasks. These tasks are simple and non-contrived, and they all concern the recovery of different types of (quantum) information. Importantly, we show that one of them can further fully characterise whether one instrument is more resourceful than another under a natural set of allowed operations.  In addition, by recognising quantum channels and POVMs as types of instruments, our resource theory can, in these special cases, be reduced to established resource theories for these two concepts, thereby also reproducing key results in these frameworks.

\section{Preliminaries} 

A quantum state, $\hat{\rho}$, is a positive semi-definite unit-trace operator: $\hat{\rho}\succeq0$ and $\tr(\hat{\rho})=1$. A POVM is given by $\{E_a\}_a$ with $E_a\succeq0$ and $\sum_aE_a=\mathds{1}$. For any initial state $\hat{\rho}$, measuring the POVM $\{E_a\}_a$ yields the classical outcome $a$ with probability $p(a)=\tr(E_a\hat{\rho})$. After measurement, the state is updated into a post-measurement state conditioned on the outcome. To also describe this updated state, we need to go beyond POVMs and utilise quantum instruments. Formally, an instrument ${\bm\mcE}\coloneqq\{\mcE_a\}_a$ is a set of completely-positive trace-non-increasing linear maps with the property that $\sum_a \mcE_a $ is a channel (namely, a completely-positive trace-preserving linear map). For clarity, we use objects with ``hat'' (such as $\hat{\rho}$ and $\hat{\mcE}$) to denote normalised states or channels, and leave off hats for unnormalised objects. For any initial state $\hat{\rho}$, the instrument ${\bm\mcE}$ describes the process that generates the classical outcome $a$ and the quantum output $\mathcal{E}_a(\hat{\rho})/\tr[\mathcal{E}_a(\hat{\rho})]$ with the probability $p(a)=\tr[\mathcal{E}_a(\hat{\rho})]$. 

Quantum instruments encompass channels, measurements and states as special cases. Consider an instrument mapping a quantum state on system $A$ to a (sub-normalised) quantum state on system $B$. We write it as ${\bm\mcE}^{\AtoB{A}{B}}\coloneqq\{\mathcal{E}_a^{\AtoB{A}{B}}\}_{a=1}^{|\mathcal{A}|}$ with $|\mathcal{A}|$ possible classical outcome (when necessary, we use superscripts to clarify which systems the object is transforming to and from). When $|\mathcal{A}|=1$ (trivial classical outcome), ${\bm\mcE}^{\AtoB{A}{B}}$ is identical to a channel. When $d_B\coloneqq{\rm dim}(B)=1$ (trivial quantum output), each $\mcE_a^{\AtoB{A}{B}}$ is equivalent to a mapping $(\cdot)^A\mapsto\tr[E_a^A(\cdot)^A]$ for some POVM $\{E_a^A\}_a$ acting on $A$. Finally, when $d_A\coloneqq{\rm dim}(A)=1$ (trivial quantum input), ${\bm\mcE}^{\AtoB{A}{B}}$ becomes an ($|\mathcal{A}|$ element) ensemble of states on $B$. Hence, as we will also later encounter, general findings for instruments can provide insights also into channels, POVMs and states.

It is worth noting again that moving from POVMs and channels to instruments enables us to model the whole measurement process, as the formalism of instruments naturally includes both the classical outcome as well as the quantum output. At the level of POVMs, one can only extract the classical outcome and the associated probabilities, without any knowledge of the updated quantum state. On the other hand, one can write the average performance of a measurement process as a channel. For instance, for a POVM $\{E_a^A\}_{a=1}^{|\mathcal{A}|}$ acting on the system $A$, the channel
$
(\cdot)^A\mapsto\sum_a\sqrt{E_a^A}\,(\cdot)^A\,\sqrt{E_a^A}
$
outputs the quantum state averaged for individual measurement outcomes. 
Still, it cannot tell us the post-measurement state after {\em each} measurement round with a specified classical outcome. It turns out that the simplest possible way to model the full measurement process is through the notion of instruments.

\section{Interactive Instruments and Their Quantification}
\subsection{Defining interactive instruments} 
We are interested in characterising the extent to which an instrument is able to coherently interact with the state that it is applied to. To this end, we begin with defining instruments that cannot  interact with the state.  A non-interactive instrument must generate both the classical and quantum output independently of the state. Hence, the measurement process first discards the state. Then, it selects the classical outcome $a$ based on sampling from an internal random variable $\lambda$ via a distribution $q_\lambda$ and a probabilistic rule $p(a|\lambda)$. Finally, based on the sample $\lambda$, a new quantum state $\hat{\sigma}_{\lambda}$ is internally prepared and emitted in the role of the quantum output. We label {\em non-interactive instrument} by ${\bm{\mathcal{L}}}=\{\mathcal{L}_a\}_a$  and formally write them as
\begin{equation}\label{Eq:trivial instrument}
\mathcal{L}_a(\cdot)= \sum_\lambda q_\lambda p(a|\lambda) \hat{\sigma}_\lambda \tr\left(\cdot\right)\quad\forall\;a,
\end{equation} 
where we denote by ``$(\cdot)$'' the arbitrary input of the instrument. Notably, by defining the probability $r(a)\coloneqq\sum_\lambda q_\lambda p(a|\lambda)$ and the state $\hat{\tau}_a\coloneqq\sum_\lambda q_\lambda p(a|\lambda) \hat{\sigma}_\lambda/r(a)$, non-interactive instruments can be equivalently expressed as 
\begin{align}
\mathcal{L}_a(\cdot)= r(a)\hat{\tau}_a\tr\left(\cdot\right)\quad\forall\,a.
\end{align} 
See Fig.~\ref{Fig:NI} for the schematic illustration. We label by ${\bf NI}$ the set of non-interactive instruments and therefore consider any instrument ${\bm\mcE}\notin{\bf NI}$ to be interactive.

\subsection{Quantifying interactive instruments} 

We now define a quantifier of the interactivity of instruments. Since the set $\CY{\bf NI}$ is convex, one can construct this measure following a standard recipe for convex resource theories~\cite{ChitambarRMP2019}. Specifically, interactiveness can be faithfully quantified by  the \CY{\em interactive instrument robustness}   (see also Ref.~\cite{Ghai2025}),
\begin{align}\label{Eq:robustness def}
R({\bm\mcE})\coloneqq\min\left\{t\ge0\,\Big|\,\frac{{\bm\mcE}+t{\bm\mcN}}{1+t}\in\CY{\bf NI}\right\},
\end{align}
where we define $({\bm\mcE}+t{\bm\mcN})/(1+t)\coloneqq\{({\mcE_a}+t{\mcN_a})/(1+t)\}_a$, and the minimisaiton runs over all instruments ${\bm\mcN}$ and scalars $t\geq 0$.  We can immediately see that $R({\bm\mcE})=0$ if and only if ${\bm\mcE}\in{\bf NI}$.
Moreover, the larger is $R({\bm\mcE})$, the more ``noise'' (represented by ${\bm\mcN}$) one needs to add to make ${\bm\mcE}$ non-interactive.

Due to the structure of non-interactive instruments, we can represent $R({\bm\mcE})$ as {\em semi-definite program} (SDP) (see, e.g., Refs.~\cite{Boyd-book,SDP-textbook}). This makes it efficiently computable. 
Through the duality theory of SDPs, we arrive at the following characterisation (we discuss the SDP formulation and give the full proof in Appendix~\ref{App:SDP for robustness}):

\begin{result}
For any instrument ${\bm\mcE}^{\AtoB{A}{B}}$, the interactive instrument robustness is given by the SDP
\begin{equation}
\begin{split}\label{Eq: main dual}
R({\bm\mcE}^{\AtoB{A}{B}})=  \max_{\{\omega^{AB}_a\}_a}\quad& d_A\sum_a\tr\left( \omega_a^{AB}\left({\rm id}^A\otimes\mcE_a^{\AtoB{A}{B}}\right)(\Phi_+^{AA})\right)-1\\
  \text{s.t.} \quad& \tr_A(\omega_a^{AB})= \mathds{1}^B, \qquad \omega^{AB}_a \succeq 0,
\end{split}
\end{equation}
where ${\rm id}^X$ is the identity channel acting on a system $X$, and $\Phi_+^{AA}\coloneqq\ketbra{\Phi_+^{AA}}$ with $\ket{\Phi_+^{AA}}\coloneqq\left(1/\sqrt{d_A}\right)\sum_{i=0}^{d_A-1}\ket{i}^A\otimes\ket{i}^{A}$ is a maximally entangled state shared by two copies of $A$.
\end{result}
By using this  SDP representation, we will show that  $R({\bm\mcE})$ is not just an abstract quantifier of interactive instruments but that it admits  relevant  operational interpretations.

\begin{figure}
	\includegraphics[width=0.35\textwidth]{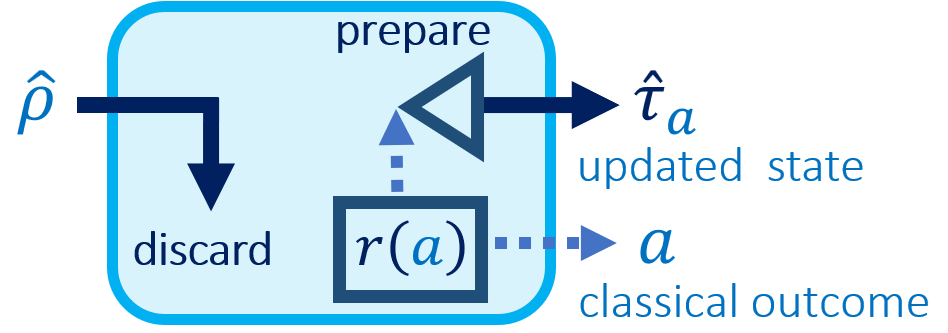}
	\caption{
	A non-interactive instrument that discards the input and prepares both the classical outcome $a$ and updated state $\hat{\tau}_a$ based on a pre-fixed distribution $r(a)$.
     	\label{Fig:NI}}
	\end{figure}

\begin{figure*}
    	\includegraphics[width=0.9\textwidth]{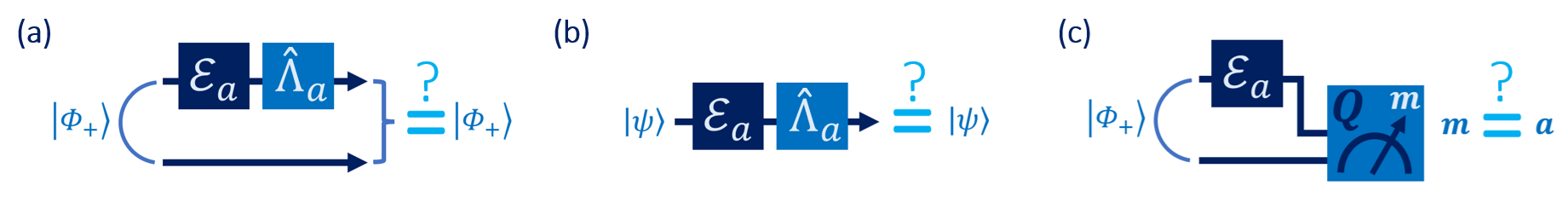}
	\caption{
	Schematic representations of the tasks for (a) measuring maximally entangled fraction, (b) benchmarking average quantum communication fidelity, and (c) entanglement-assisted unambiguous discrimination. Through the interactive instrument robustness, these three tasks are shown equivalent. 	\label{Fig:tasks}}
	\end{figure*}

\subsection{Interpretation 1: Maximally entangled fraction} 

The SDP in \eqref{Eq: main dual} enables a simple operational interpretation for the interactive instrument robustness. It corresponds to the largest average maximally entangled fraction that can be recovered after applying the instrument to one share of $\ket{\Phi_+^{AA}}$. To see this, consider that we are given $\ket{\Phi_+^{AA}}$ and measure half of it using the instrument ${\bm\mcE}^{\AtoB{A}{B}}$. The classical outcome is used to select a channel $\hat{\Lambda}_a^{\AtoB{B}{A}}$ whose job is to restore the state as close to $\ket{\Phi_+^{AA}}$ as possible. 
The average fidelity of this restoration is
\begin{equation}\label{Eq:F_+ def}
F_+({\bm\mcE}^{\AtoB{A}{B}})\coloneqq\max_{\{\hat{\Lambda}_a^{\AtoB{B}{A}}\}_a}\sum_a \bracket{\Phi_+^{AA}}{[{\rm id}^A\otimes(\hat{\Lambda}_a\circ\mcE_a)^A](\Phi_+^{AA})}{\Phi_+^{AA}},
\end{equation}
which maximises over all channels $\hat{\Lambda}_a^{\AtoB{B}{A}}$.
See also Fig.~\ref{Fig:tasks}(a).
Below, we show the relation between $F_+({\bm\mcE}^{\AtoB{A}{B}})$ and the interactive instrument robustness.\\

\begin{result}\label{Result:F_+}
For any instrument ${\bm\mcE}^{\AtoB{A}{B}}$ it holds that
\begin{align}
1+R({\bm\mcE}^{\AtoB{A}{B}}) = d_A^2\ F_+({\bm\mcE}^{\AtoB{A}{B}}).
\end{align}
\end{result}
\textit{Proof.} To see how \eqref{Eq:F_+ def} is related to the dual problem \eqref{Eq: main dual}, note that we can interpret $\omega^{AB}_a$ as the Choi operator of the {\em adjoint} of a channel. 
More precisely, each $\omega_a^{AB}\succeq0$ with $\tr_A(\omega_a^{AB})= \mathds{1}^B$ can be written as $\omega_a^{AB}=({\rm id}^B\otimes\Gamma^{\AtoB{A}{B}})(d_A\Phi_+^{AA})$ for some completely-positive {\em unital} linear map $\Gamma^{\AtoB{A}{B}}$ from $A$ to $B$.
This can be seen by tracing out the first copy of $A$ and noticing that $\Gamma^{\AtoB{A}{B}}(\mathds{1}_A)=\mathds{1}_B$.
As $(\Gamma^{\AtoB{A}{B}})^\dagger$ is trace-preserving, we can thus write $\Gamma^{\AtoB{A}{B}}=(\hat{\Lambda}_a^{\AtoB{B}{A}})^\dagger$ for some channel $\hat{\Lambda}_a^{\AtoB{B}{A}}$ mapping from $B$ to $A$.
Substituting this into \eqref{Eq: main dual} and using the definition of the adjoint map thus completes the proof.
\hfill$\square$\\

Hence, the more interactive an instrument is, the better it allows for the maximally entangled state to be locally restored after measurement. Moreover, non-interactive instruments cannot exceed $F_{+}({\bm{\mathcal{L}}}^{\AtoB{A}{B}})=1/{d_A^2}$. One can therefore interpret  $1+R({\bm\mcE}^{\AtoB{A}{B}})$ as the task performance ratio between a given instrument and arbitrary non-interactive ones. 
Finally, Result~\ref{Result:F_+} implies 
\begin{align}
R({\bm\mcE}^{\AtoB{A}{B}})\le d_A^2-1,
\end{align}
and the maximum can be achieved when $\mcE^{\AtoB{A}{B}}_a=p_a\mathcal{U}_a^{\AtoB{A}{B}}$ for some probability $p_a$ and reversible channels $\mathcal{U}_a^{\AtoB{A}{B}}$. These maximally interactive instruments are those that can be undone conditioned on classical outcomes.

\subsection{Interpretation 2: Fidelity of quantum communication}

We now present a second operational interpretation of the interactive instrument robustness. This time, it does not rely on entangled states but only on randomly selected single systems. Consider that we sample a random pure state $\psi^A\coloneqq\ketbra{\psi^A}$ and measure it using the  instrument ${\bm\mcE}^{\AtoB{A}{B}}$. 
The classical outcome, $a$, is used to select a channel $\hat{\Lambda}_a^{\AtoB{B}{A}}$ whose job is to minimise the disturbance in the quantum state caused by the instrument. Specifically, we want to maximise the fidelity between the pre- and post-measurement quantum state. On average, this fidelity corresponds to [see also Fig.~\ref{Fig:tasks}(b)]
\begin{equation}\label{Eq:ave fidelity main}
F_{\rm ave}({\bm\mcE}^{\AtoB{A}{B}})\coloneqq\max_{\{\hat{\Lambda}_a^{\AtoB{B}{A}}\}_a}\int\!\! d\psi\ \bracket{\psi^A}{\sum_a (\hat{\Lambda}_a\circ \mcE_a)^A(\psi^A)}{\psi^A},
\end{equation}
where the maximisation is over the channels $\hat{\Lambda}_a^{\AtoB{B}{A}}$. In Appendix~\ref{App:Proof-Result:F_ave R} we prove that this is in one-to-one correspondence with the interactive instrument robustness through the following result.

\begin{result}\label{Result:F_ave R}
For any instrument ${\bm\mcE}^{\AtoB{A}{B}}$, it holds that
\begin{equation}
F_{\rm ave}({\bm\mcE}^{\AtoB{A}{B}})=\frac{1}{d_A}\left(1+\frac{R\left({\bm\mcE}^{\AtoB{A}{B}}\right)}{d_A+1}\right).
\end{equation}
\end{result}
Note that for non-interactive instruments, we obtain the trivial fidelity $F_{\rm ave}({\bm\mcE}^{\AtoB{A}{B}})=1/d_A$. Furthermore, by combining Result~\ref{Result:F_+} and Result~\ref{Result:F_ave R}, we obtain
\begin{equation}\label{Eq: relating F_+ and F_ave 001}
F_{\rm ave}({\bm\mcE}^{\AtoB{A}{B}})=\frac{1+d_AF_+\left({\bm\mcE}^{\AtoB{A}{B}}\right)}{d_A+1}.
\end{equation}
This shows that while the two tasks are operationally distinct, their performance  metrics are in one-to-one correspondence.

\subsection{Interpretation 3: Entanglement-assisted unambiguous discrimination} 

We now discuss a third operational interpretation, which we will afterwards extend further into the complete resource theory of interactive instruments.

Consider that a sender and a receiver share the maximally entangled state $\ket{\Phi_+^{AA}}$. The sender measures using a  given instrument ${\bm\mcE}^{\AtoB{A}{B}}=\{\mcE_a^{\AtoB{A}{B}}\}_{a=1}^{|\mathcal{A}|}$, with $|\mathcal{A}|$ possible classical outcomes. She keeps the classical outcome while relaying the quantum output to the receiver.  Conditioned on the outcome $a$, which is not accessible to the receiver, the bipartite state is $\hat{\rho}_a^{AB}\coloneqq({\rm id}^A\otimes\mathcal{E}_a^{\AtoB{A}{B}})(\Phi_+^{AA})/q(a)$ with the probability $q(a)\coloneqq\tr\left[({\rm id}^A\otimes\mathcal{E}_a^{\AtoB{A}{B}})(\Phi_+^{AA})\right]$. The receiver's goal is to learn the value of $a$ with the largest possible probability, without ever declaring a wrong value. 
In other words, the classical outcome must be unambiguously discriminated by the receiver with a probability as high as possible. To that end, the receiver jointly measures system  $AB$ with a POVM ${\bm Q}^{AB}\coloneq\{Q_{m}^{AB}\}$ with $m=1,...,|\mathcal{A}|,\emptyset$ that has $|\mathcal{A}|+1$ possible outcomes. The extra outcome ($\emptyset$) corresponds to an inconclusive round, which the receiver simply discards.
See Fig.~\ref{Fig:tasks}(c) for illustration.
Hence, the probability of success is given by the standard figure-of-merit for unambiguous discrimination
\begin{align}\label{Eq: P_succ}
P_{\rm succ}({\bm\mcE}^{\AtoB{A}{B}},{\bm Q}^{AB})=\sum_{a=1}^{|\mathcal{A}|}\tr\left[Q_a^{AB}({\rm id}^A\otimes\mcE_a^{\AtoB{A}{B}})(\Phi_+^{AA})\right],
\end{align} 
while the unambiguity means $\tr\left[Q_{m}^{AB}({\rm id}^A\otimes\mcE_a^{\AtoB{A}{B}})(\Phi_+^{AA})\right]=0$ for any $a\neq m \in\{1,\ldots,|\mathcal{A}|\}$. Then, as detailed in Appendix~\ref{App:Proof-Result: three operational meanings}, $P_{\rm succ}$ is related to the interactive instrument robustness as follows.\\

\begin{result}\label{Result: three operational meanings}
For every instrument ${\bm\mcE}^{\AtoB{A}{B}}$ it holds that
\begin{align}\label{Eq: equivalence relation between three tasks}
1+R\left({\bm\mcE}^{\AtoB{A}{B}}\right)=\max_{{\bm Q}^{AB}}\frac{P_{\rm succ}({\bm\mcE}^{\AtoB{A}{B}},{\bm Q}^{AB})}{\max_{{\bm{\mathcal{L}}}^{\AtoB{A}{B}}}P_{\rm succ}({\bm{\mathcal{L}}}^{\AtoB{A}{B}},{\bm Q}^{AB})},
\end{align}
where $\max_{\bf Q}$ is over all POVMs  and $\max_{{\bm{\mathcal{L}}}}$ is over all non-interactive instruments with the same number of classical outcomes.
\end{result}

Hence, the above result has uncovered a third operational meaning of $R({\bm\mcE})$, but in contrast to the previous two, as the task is now focused on the recovery of classical information. We remark that for a given POVM \CY{${\bm Q}^{AB}$}, the  success probability for non-interactive instruments can be explicitly evaluated by  
\begin{align}
P_{\rm succ}({\bm{\mathcal{L}}}^{\AtoB{A}{B}},{\bm Q}^{AB})=\max_{a\in\{1,...,|\mathcal{A}|\}}\frac{\norm{\tr_A\left(Q_a^{AB}\right)}_\infty}{d_A}.
\end{align}

\section{Resource characterisation} 

\subsection{Defining allowed operations}\label{Sec:Allowed Operation}

So far, we have solved the membership problem of the set of non-interactive instruments and uncovered several operational meanings of the interactive instrument robustness. A natural next step is to introduce a reasonable set of allowed operations that can be performed on instruments. By identifying the conditions under which one instrument can be converted to another under these operations, we obtain a complete resource theory.

We first motivate the selection of allowed operations. It consists in quantum pre-processing, classical post-processing and quantum post-processing of a given instrument. Specifically, consider a generic input state $\hat{\rho}_{\rm in}$.
For a given instrument ${\bm\mcE}$, the agent can sample $\lambda$ from a probability distribution  $q(\lambda)$ and then apply some  channel $\hat{\mathcal{P}}_\lambda$. This can be viewed as quantum pre-processing. Then, the agent may apply the given instrument, generating a classical outcome $k$ and a quantum output $\hat{\rho}_{k|\lambda} \coloneqq (\mcE_k\circ\hat{\mathcal{P}}_\lambda)(\hat{\rho}_{\rm in})/Q(k|\lambda)$ with the probability $Q(k|\lambda)\coloneqq\tr[(\mcE_k\circ\hat{\mathcal{P}}_\lambda)(\hat{\rho}_{\rm in})]$.
At this stage, the agent is allowed to classically process $k$ to generate the final classical outcome. This means employing some probability distribution $P(a|k,\lambda)$. Finally, the agent is allowed to perform quantum post-processing by  applying to $\hat{\rho}_{k|\lambda}$ some channel $\hat{\mathcal{D}}_{k|\lambda}$ conditioned on $k$ and $\lambda$. 
In summary, the allowed operations, when conditioned on outcome $a$, correspond to transformations $\hat{\rho}_\text{in}\rightarrow \mathcal{N}_a(\hat{\rho}_{\rm in})$  where (see Fig.~\ref{Fig:AO} for illustration)
\begin{align}\label{free op}
\mathcal{N}_a(\hat{\rho}_{\rm in})=\sum_{\lambda,k}q(\lambda)P(a|k,\lambda)(\hat{\mathcal{D}}_{k|\lambda}\circ\mcE_k\circ\hat{\mathcal{P}}_\lambda)(\rho_{\rm in}).
\end{align}
For a given $\bm \mcE$, if the above transformation is possible for arbitrary input states, we write ${\bm\mcN} = \Pi({\bm\mcE})$, where $\Pi$ is the appropriate supermap for the conversion.
We label the set of all allowed operations by ${\bf AO}$.

We now discuss the mathematical property of the set ${\bf AO}$.
First, as one can directly observe, ${\bf AO}$ is by definition convex.
In fact, by considering fixed input and output dimensions, it is also compact (in a topology induced by operator norm; see Lemma~\ref{lemma:compact} in Appendix~\ref{App:Proof-Result:conversion iff conditions} for details).
Physically, this means that if a transformation can be approximated by allowed operations to {\em arbitrary precision}, then such a manipulation is also allowed.

\begin{figure}
	\includegraphics[width=0.3\textwidth]{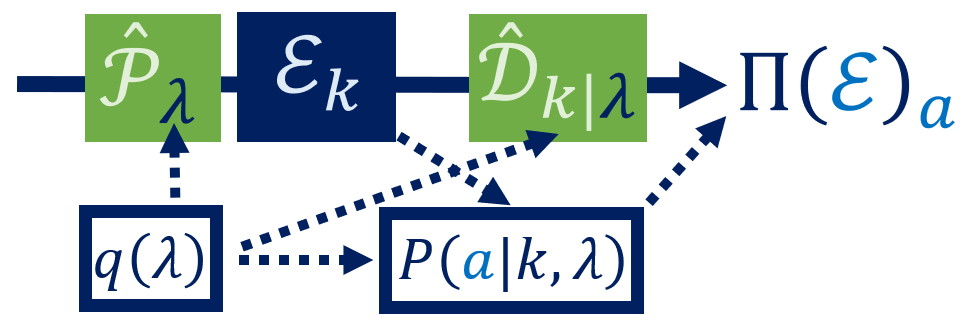}
	\caption{
	\CY{An allowed operation $\Pi$ acting on ${\bm\mcE}$. Solid/dashed lines represent quantum/classical processing.} 	\label{Fig:AO}}
	\end{figure}

To examine the relation between non-interactive instruments ${\bf NI}$ and allowed operations ${\bf AO}$, it is straightforward to see that 
\begin{align}
\Pi({\bm{\mathcal{L}}})\in{\bf NI}\quad\forall\,{\bm{\mathcal{L}}}\in{\bf NI}\;\&\;\Pi\in{\bf AO}.
\end{align} Thus, the allowed operations cannot convert resourceless objects to resourceful objects.  If there is some $\Pi\in{\bf AO}$ such that $\Pi({\bm\mcE})={\bm\mcN}$, then we say that one can simulate the instrument ${\bm\mcN}$ with the instrument  ${\bm\mcE}$, i.e.,~one can view ${\bm\mcE}$ as not less resourceful than ${\bm\mcN}$. A rather expected consequence of this is that the interactive instrument robustness is non-increasing under allowed operations (see also Ref.~\cite{Ghai2025} for a more complete treatment on quantifying instrument-based resources): 
\begin{align}\label{Eq: non-increasing property}
R[\Pi({\bm\mcE})]\le R({\bm\mcE})\quad\forall\;{\bm\mcE}\;\&\;\Pi\in {\bf AO}.
\end{align}
\noindent{\em Proof of \eqref{Eq: non-increasing property}.}
Throughout the proof, we ignore the superscripts, and we use the notation ``$\Pi({\bm\mcE})_a$'' to denote the $a$-th element of the instrument $\Pi({\bm\mcE})$, and we define $p{\bm\mcE}+(1-p){\bm\mcN}\coloneqq\{p\mcE_a+(1-p)\mcN_a\}_a$.
From 
\eqref{Eq:robustness def}, 
we have
\begin{align}\nonumber
R({\bm\mcE})\coloneqq  \quad \min &\quad t\\
  \text{s.t.} &\quad \frac{{\mcE}_a(\cdot)+t \mathcal{N}_a(\cdot)}{1+t}= r(a)\hat{\tau}_a \tr\left(\cdot\right)\quad\forall\,a.
\end{align}
Write $t_{\rm opt}=R({\bm\mcE})$.
By applying the allowed operation $\Pi$ to the above constraint and using the fact that $\Pi$ must map a non-interactive instrument to some non-interactive instrument, we conclude that $\Pi\left[({\bm\mcE}+t_{\rm opt} {\bm\mcN})/(1+t_{\rm opt})\right]$ again form a non-interactive instrument.
Hence, $t_{\rm opt}$ is a feasible (not necessarily optimal) solution to the following minimisation, meaning that $t_{\rm opt}$ is lower bounded by it:
\begin{align}\nonumber
\quad \min &\quad t\\
  \text{s.t.} &\quad \Pi\left(\frac{{\bm\mcE}+t {\bm\mcN}}{1+t}\right)_a(\cdot)= r'(a)\hat{\tau}'_a \tr\left(\cdot\right)\quad\forall\,a.
\end{align}
Since $\Pi$ is linear, we can write 
\begin{align}
\Pi\left(\frac{{\bm\mcE}+t {\bm\mcN}}{1+t}\right)=\frac{\Pi({\bm\mcE})+t \Pi({\bm\mcN})}{1+t}, 
\end{align}
and $\Pi({\bm\mcN})$ is still a valid instrument.
By enlarging the optimisation range, the above minimisation is sub-optimal to (and hence lower bounded by) the following minimisation, which is exactly $R[\Pi({\bm\mcE})]$:
\begin{align}\nonumber
R[\Pi({\bm\mcE})]\coloneqq\quad \min &\quad t\\
  \text{s.t.} &\quad \frac{\Pi({\bm\mcE})_a(\cdot)+t \mathcal{N}'_a(\cdot)}{1+t}= r'(a)\hat{\tau}'_a \tr\left(\cdot\right)\;\forall\,a,
\end{align}
where the above minimisation allows all possible instruments $\mcN_a'$ rather than ones of the form $\Pi({\bm\mcN})_a$.
Combining everything, we have thus proved that $R({\bm\mcE})=t_{\rm opt}\ge R[\Pi({\bm\mcE})]$, which completes the proof.
\hfill$\square$

\subsection{Characterising instrument conversions}

We are now ready to give a full characterisation of instrument conversions under the allowed operations. To this end, we return to the third interpretation of the interactive instrument robustness and the quantity in \eqref{Eq: P_succ}. For a given instrument ${\bm\mcE}^{\AtoB{A}{B}}=\{\mcE_a^{\AtoB{A}{B}}\}_{a=1}^{|\mathcal{A}|}$ and POVM ${\bm Q}^{AB}\coloneq\{Q_{m}^{AB}\}$ with $m=1,...,|\mathcal{A}|,\emptyset$, the largest achievable success probability under the allowed operations is
\begin{align}
P_{\rm max}({\bm\mcE}^{\AtoB{A}{B}},{\bm Q}^{AB})\coloneqq\max_{\Pi\in{\bf AO}^{\AtoB{A}{B}}}P_{\rm succ}\left[\Pi({\bm\mcE}^{\AtoB{A}{B}}),{\bm Q}^{AB}\right],
\end{align}
which maximises over allowed operations that will map ${\bm\mcE}^{\AtoB{A}{B}}$ to some instrument that is again from $A$ to $B$ (denoted by ${\bf AO}^{\AtoB{A}{B}}$).
In Appendix~\ref{App:Proof-Result:conversion iff conditions}, we prove that this quantity characterises instrument conversions as follows.

\begin{result}\label{Result:conversion iff conditions}
For instruments ${\bm\mcE}^{\AtoB{A}{B}}$ and ${\bm\mcN}^{\AtoB{A}{B}}$, there is an allowed operation $\Pi\in{\bf AO}^{\AtoB{A}{B}}$ achieving $\Pi({\bm\mcE}^{\AtoB{A}{B}})={\bm\mcN}^{\AtoB{A}{B}}$ if and only if  
\begin{align}\label{Eq:iff}
 P_{\rm max}({\bm\mcE}^{\AtoB{A}{B}},{\bm Q}^{AB})\ge P_{\rm max}({\bm\mcN}^{\AtoB{A}{B}},{\bm Q}^{AB}) \quad \forall\;{\bm Q}^{AB}.
\end{align}
\end{result}
In other words, ${\bm\mcE}^{\AtoB{A}{B}}$ can simulate ${\bm\mcN}^{\AtoB{A}{B}}$ if and only if the former outperforms the latter for any choice of \CY{POVM ${\bm Q}^{AB}$} in the \CYtwo{unambiguous discrimination} task associated with \eqref{Eq: P_succ}.

\section{Examples and Special Cases}
\subsection{Special case: channels}

Since quantum instruments generalise both channels and POVMs, we discuss how our resource theory applies to these special cases. Beginning with the former, consider that our instrument only has one possible classical outcome (i.e.~$|\mathcal{A}|=1$), which reduces it to a channel. Then, the non-interactive instrument defined in \eqref{Eq:trivial instrument} becomes a state-preparation channel, namely $(\cdot)\mapsto\hat{\rho}\tr(\cdot)$, for some fixed state $\hat{\rho}$. The set of non-interactive instruments ${\bf NI}$ becomes the set of all state-preparation channels, and the interactive instrument robustness simplifies to have no dependency on the classical outcome. This is precisely what is encountered in the so-called resource theory of communication developed in Ref.~\cite{Takagi2020PRL}.  

Furthermore, the quantity $F_+(\hat{\mathcal{E}}^{\AtoB{A}{B}})$ in \eqref{Eq:F_+ def} becomes the maximally entangled fraction of the channel's Choi state \cite{Horodecki1999,Konig2009ITIT}. 
This is also related to the min-entropy of this Choi state.
To see this, recall from Definition 1 in Ref.~\cite{Konig2009ITIT} that the min-entropy (of $A$ conditioned on $B$) of $\hat{\rho}^{AB}$ reads
\begin{align}
H_{\rm min}(A|B)_{\hat{\rho}^{AB}}\coloneqq-\min_{\hat{\sigma}^B}D_{\rm max}(\hat{\rho}^{AB}\,\|\,\mathds{1}^A\otimes\hat{\sigma}^B),
\end{align}
which minimises over generic states $\hat{\sigma}^B$ in $B$. Moreover, for two states $\rho,\eta$, their max-relative entropy~\cite{Datta2009ITIT} reads
\begin{align}
D_{\rm max}(\rho\,\|\,\eta)\coloneqq\log_2\min\{\lambda\ge0\,|\,\rho\le\lambda\eta\}.
\end{align}
This thus implies, for $\hat{\rho}^{AB}=({\rm id}^A\otimes\hat{\mcE}^{\AtoB{A}{B}})(\Phi_+^{AA})$,
\begin{align}
2^{-H_{\rm min}(A|B)_{\hat{\rho}^{AB}}}&=\frac{1}{d_A}\min\left\{\lambda\ge0\,\Big|\,\hat{\rho}^{AB}\le\lambda\frac{\mathds{1}^A}{d_A}\otimes\hat{\sigma}^B,\hat{\sigma}:{\rm state}\right\}\nonumber\\
&=\frac{1+R(\hat{\mcE}^{\AtoB{A}{B}})}{d_A}=d_A\ F_+({\bm\mcE}^{\AtoB{A}{B}}),
\end{align}
where we have used Result~\ref{Result:F_+} in the last line.
This is a well-known relation for establishing the operational meaning of min-entropy~\cite{Konig2009ITIT}. Moreover, using \eqref{Eq: relating F_+ and F_ave 001}, we further see that the maximally entangled fraction of the Choi state ($F_+$) is determined by the quantity 
\begin{align}
F_{\rm ave}(\hat{\mcE}^{\AtoB{A}{B}})\coloneqq\max_{\hat{\Lambda}^{\AtoB{B}{A}}}\int d\psi \bracket{\psi^A}{(\hat{\Lambda}\circ \hat{\mcE})^A(\psi^A)}{\psi^A},
\end{align} 
which is the optimal average fidelity of $\hat{\mcE}^{\AtoB{A}{B}}$'s ability to send quantum information. When $d_A=d_B$, this is exactly the well-known relation between Choi state's fully entangled fraction and average fidelity reported in Refs.~\cite{Horodecki1999,Nielsen2002PLA,Buscemi2010ITIT}.


\subsection{Special case: POVMs} 

Instruments reduce to POVMs if we discard the quantum output, i.e.~select $d_B=1$. This is expressed as \mbox{$\mcE_a^{\AtoB{A}{B}}(\cdot)\coloneqq\tr[E_a^A(\cdot)^A]$} 
for some POVM ${\bm E}^A=\{E_a^A\}_a$.
Then, a non-interactive instrument takes the form
$
\sum_\lambda q_\lambda P(a|\lambda)\tr(\cdot)=\tr\left[\left(\sum_\lambda q_\lambda P(a|\lambda)\mathds{1}\right)(\cdot)\right]$ $\forall\,a
$, which means that the system is discarded and the outcome selected by an internal random process. This is precisely the type of POVMs considered in the resource theory of measurements introduced in Ref.~\cite{Skrzypczyk2019PRL}. The interactive instrument robustness reduces to the {\em robustness of measurements} (denoted by $R_{\rm POVM}$), which, for any POVM ${\bm M}=\{M_a\}_a$, is defined as~\cite{Skrzypczyk2019PRL}
\begin{equation}\label{Eq: robustness POVM}
\begin{split}
R_{\rm POVM}({\bm M})\coloneqq  \min &\quad t\\
  \text{s.t.} &\quad \frac{M_a+tN_a}{1+t}=r(a)\mathds{1}\quad\forall\,a,
\end{split}
\end{equation}
which minimises over the scalar $t\ge0$ and generic POVM ${\bm N}$ acting on the same system.
Now, for a given instrument ${\bm\mcE}^{\AtoB{A}{B}}$ with a trivial output space (i.e., $d_B=1$), suppose it is expressed by a POVM ${\bm E}^A=\{E_a^A\}_a$ as $\mcE^{\AtoB{A}{B}}_a(\cdot)=\tr[E_a^A(\cdot)^A]$ $\forall\,a$.
Then, we aim to prove that
\begin{align}\label{Eq:relating R and R_POVM}
R({\bm\mcE}^{\AtoB{A}{B}})=R_{\rm POVM}({\bm E}^A).
\end{align}

\noindent\textit{Proof of \eqref{Eq:relating R and R_POVM}.---} Using \eqref{Eq:robustness def}, the interactive instrument robustness $R({\bm\mcE}^{\AtoB{A}{B}})$ is the following optimisation:
\begin{equation}\label{Eq: App_B 001}
\begin{split}
\min &\quad t\\
  \text{s.t.} &\quad \frac{\tr[E_a^A(\cdot)]+t \mathcal{N}_a^{\AtoB{A}{B}}(\cdot)}{1+t}= \tr\left[r(a)\mathds{1}^A(\cdot)\right]\quad\forall\,a,
\end{split}
\end{equation}
which miminises over probability distribution $\{r(a)\}_a$, scalar $t\ge0$, and generic instrument $\{\mcN_a^{\AtoB{A}{B}}\}_a$.
Now, we note that the above constraint implies
\begin{align}
\mathcal{N}_a^{\AtoB{A}{B}}(\cdot)=\tr\left[\left(\frac{(1+t)r(a)\mathds{1}^A-E_a^A}{t}\right)(\cdot)\right]\quad\forall\,a.
\end{align}
We further note that $[(1+t)r(a)\mathds{1}^A-E_a^A]/t\succeq0$ $\forall\,a$ since $\mathcal{N}_a^{\AtoB{A}{B}}$ is completely-positive.
Also, since $E_a^A$'s form a POVM, we have that $\sum_a[(1+t)r(a)\mathds{1}^A-E_a^A]/t=\mathds{1}^A$.
This means we can write $\mathcal{N}_a^{\AtoB{A}{B}}(\cdot)=\tr[N_a^A(\cdot)]$ for some POVM $\{N_a^A\}_a$, and \eqref{Eq: App_B 001} becomes
\begin{equation}
\begin{split}
\min &\quad t\\
  \text{s.t.} &\quad \tr[\left(\frac{E_a^A+tN_a^A}{1+t}-r(a)\mathds{1}^A\right)(\cdot)]=0\quad\forall\,a.
\end{split}
\end{equation}
For an operator $\Delta$, the condition that $\tr[\Delta(\cdot)]=0$ (i.e., it is the zero map for {\em all} states) implies $\Delta=0$.
Consequently, we obtain
\begin{equation}
\begin{split}
R({\bm\mcE}^{\AtoB{A}{B}})=  \min &\quad t\\
  \text{s.t.} &\quad \frac{E_a^A+tN_a^A}{1+t}=r(a)\mathds{1}^A\quad\forall\,a,
\end{split}
\end{equation}
which is exactly $R_{\rm POVM}({\bm E}^A)$. 
This concludes the proof.
\hfill$\square$\\

As a corollary, using Result~\ref{Result:F_+} and \eqref{Eq:relating R and R_POVM}, we obtain
\begin{align}
1+R_{\rm POVM}({\bm E}^A)&=d_A^2F_+({\bm\mcE^{\AtoB{A}{B}}})\nonumber\\
&=d_A\max_{\{\hat{\rho}^A_a\}_a}\sum_a \bracket{\Phi_+^{AA}}{\left(E_a^t\otimes\hat{\rho}_a^A\right)}{\Phi_+^{AA}}\nonumber\\
&=\max_{\{\hat{\rho}_a\}_a}\sum_a\tr\left(\hat{\rho}_aE_a\right)
=\sum_a\norm{E_a}_\infty,
\end{align}
where $\norm{E_a}_\infty$ is the operator norm of $E_a$.
This thus reproduces Eq.~(7) in Ref.~\cite{Skrzypczyk2019PRL}.
Also, together with Result~\ref{Result:F_ave R}, we obtain
\begin{align}
&\frac{1}{d_A(d_A+1)}\Big(\sum_a\norm{E_a}_\infty+d_A\Big)=F_{\rm ave}({\bm\mcE}^{({\bm E})})\nonumber\\
&=\max_{\{\hat{\rho}_a\}_a}\int d\psi \sum_a\bracket{\psi}{\hat{\rho}_a^B}{\psi}\bracket{\psi}{E_a}{\psi}\nonumber\\
&=\max_{\{\ket{\phi}_a\}_a}\int d\psi \sum_a\braket{\psi}{\phi_a}\braket{\phi_a}{\psi}\bracket{\psi}{E_a}{\psi},
\end{align}
suggesting another operational way to characterise $R_{\rm POVM}$. 

Finally, to find a complete set of criteria for resource conversions, Ref.~\cite{Skrzypczyk2019PRL} introduces an infinitely large  family of state discrimination tasks. 
This can indeed be reproduced by setting $d_B=1$ in Result~\ref{Result:conversion iff conditions}.

\section{Discussion} 

We have introduced a complete resource theory of instruments capable of coherent interactions.  We showed that a natural quantifier of the resource admits three different relevant interpretations in terms of non-contrived quantum information tasks, all focused on restoring some type of information after applying the instrument. We further also showed how one of these tasks paves the way for a complete characterisation of resource conversions under a natural set of allowed operations. 
Our findings are thus complementary to the recent work Ref.~\cite{Ghai2025}, which provides a general treatment of comparing and quantifying different instrument-based quantum resources. A natural open question is then whether one can further extend our operational interpretations and the complete set of monotones to other instrument-based resources.

The strong nature of the results derived here is largely thanks to the fact that the resourceless instruments are essentially trivial from a physical point of view, because they cannot at all interact with the state to which they are applied. One might expect similarly strong relations to be rarer when considering less basic notions of resource in instruments. Nevertheless, since adopting this type of basic resource is not uncommon in resource theories, and instruments represent a generalisation of other important quantum operations, we have found that our results return as special cases several well-known results in the literature.

We conclude with a few open questions. Firstly, can our resource theory be related to other notions of instrument resources, such as incompatible instruments~\cite{Hsieh2024PRL,Ji2024PRXQ,Ghai2025,Buscemi2023,Leppajarvi2024Quantum} or projective simulable instruments \cite{Khandelwal2025PRL}? 
A direct observation is that our resource theory is defined for a single instrument, while, for instance, incompatibility of instruments is by definition a property of at least {\em a pair of} instruments. Consequently, it is non-trivial to relate our notion of interactive instruments directly to incompatible instruments, and we leave this for future research.
Secondly, we focused on exact instrument simulations, but there may be some interest in studying approximate simulations as well.
Thirdly, this work focus on the operational interpretations as well as the existence of a complete set of monotone in the present resource theory, leaving a complete understanding of interactive instruments' mathematical structures an open question.
Finally, a natural follow-up, which is left for future work, is to consider the weight-based measure in our resource theory and explore its operational significance.

\section*{acknowledgements}
We thank \CYtwo{Shao-Hua Hu, Benjamin Stratton, and Hayata Yamasaki} for the fruitful discussion. 
C.-Y.~H.~acknowledges support from the Leverhulme Trust Early Career Fellowship (``Quantum complementarity: a novel resource for quantum science and technologies'' with Grant No.~ECF-2024-310). A.T.~is supported by the Knut and Alice Wallenberg Foundation through the Wallenberg Center for Quantum Technology (WACQT), the Swedish Research Council under Contract No.~2023-03498 and  the Swedish Foundation for Strategic Research. 
H.-Y.~K.~is supported by National Science
and Technology Council (NSTC) with Grant No.~NSTC 112-2112-M-003-020-MY3.

\appendix

\section{SDP for interactive instrument robustness}\label{App:SDP for robustness}
We first notice that for a trivial instrument ${\bm{\mathcal{L}}}^{\AtoB{A}{B}}$ with $\mathcal{L}_a^{\AtoB{A}{B}}(\cdot)=r(a)\hat{\tau}_a^B\tr(\cdot)$ for some probability $r(a)$ and state $\hat{\tau}_a^B$, by letting $\tau_a^B\coloneqq r(a)\hat{\tau}_a^B$, its {\em Choi representation}~\cite{Choi1975,Jamiolkowski1972} reads
\begin{equation}\label{Eq: trivial instrument Choi}
\left({\rm id}^A\otimes\mathcal{L}_a^{\AtoB{A}{B}}\right)(\Phi_+^{AA})=\frac{\mathds{1}^A}{d_A}\otimes\tau^B_a
\end{equation}
with $\tau_a^B\succeq0$ and $\sum_a\tr(\tau_a^B)=1$.
For a generic instrument ${\bm\mcN}^{\AtoB{A}{B}}$ used in \eqref{Eq:robustness def}, define the new variables \mbox{$\mu_a^{AB}\coloneqq t\left({\rm id}^A\otimes\mathcal{N}_a^{\AtoB{A}{B}}\right)(\Phi_+^{AA})$} and $\kappa_a^B\coloneqq(1+t)r(a)\hat{\tau}_a^B$.
Then, the non-trivial instrument robustness $R({\bm\mcE}^{\AtoB{A}{B}})$ [\eqref{Eq:robustness def}] can be rewritten into:
\begin{equation}
\begin{split}\label{Eq:robustness primal APP}
\min&\quad \sum_a\tr(\kappa_a^B)-1\\
  \text{s.t.} &\quad \left({\rm id}^A\otimes\mcE_a^{\AtoB{A}{B}}\right)(\Phi_+^{AA})+\mu^{AB}_a=\frac{\mathds{1}^A}{d_A}\otimes\kappa_a^B,\\
& \quad t\ge0,\mu_a^{AB}\succeq 0,\kappa_a^B\succeq 0\quad\forall\,a,\\
&\quad\sum_a\tr(\kappa_a^B)=1+t.
\end{split}
\end{equation}
Note that the constraint $\sum_a\tr\left(\mu^{AB}_a\right)=t$ (which is the condition ensuring $\mu^{AB}_a$'s are Choi operators of a valid instrument) is guaranteed by $\sum_a\tr(\kappa_a^B)=1+t$ and is thus redundant.
Hence, as $\left({\rm id}^A\otimes\mcE_a^{\AtoB{A}{B}}\right)(\Phi_+^{AA})+\mu^{AB}_a=(\mathds{1}^A/d_A)\otimes\kappa_a^B$ for some $\mu^{AB}_a\succeq0$ if and only if $\left({\rm id}^A\otimes\mcE_a^{\AtoB{A}{B}}\right)(\Phi_+^{AA})\preceq(\mathds{1}^A/d_A)\otimes\kappa_a^B$, \eqref{Eq:robustness primal APP} is equal to
\begin{equation}
\begin{split}
\min&\quad \sum_a\tr(\kappa_a^B)-1\\
  \text{s.t.} &\quad \left({\rm id}^A\otimes\mcE_a^{\AtoB{A}{B}}\right)(\Phi_+^{AA})\preceq\frac{\mathds{1}^A}{d_A}\otimes\kappa_a^B,\\
& \quad t\ge0,\kappa_a^B\succeq 0\quad\forall\,a,\\
&\quad\sum_a\tr(\kappa_a^B)=1+t.
\end{split}
\end{equation}
Then, the constraint $\sum_a\tr(\kappa_a^B)=1+t$ becomes redundant, and hence we can simplify it into
\begin{equation}
\begin{split}\label{Eq:primal APP}
R({\bm\mcE}^{\AtoB{A}{B}})=  
\min&\quad \sum_a\tr(\kappa_a^B)-1\\
  \text{s.t.} &\quad 
  {\mathcal{J}(\mcE_a)^{AB}}\preceq\frac{\mathds{1}^A}{d_A}\otimes\kappa_a^B,\\
& \quad \kappa_a^B\succeq 0\quad\forall\,a,
\end{split}
\end{equation}
where we write ${\bm\mcE}^{\AtoB{A}{B}}$'s Choi operators~\cite{Choi1975,Jamiolkowski1972} as 
\begin{align}\label{Eq:Choi operators}
\mathcal{J}(\mcE_a)^{AB}\coloneqq\left({\rm id}^A\otimes\mcE_a^{\AtoB{A}{B}}\right)(\Phi_+^{AA}).
\end{align}
This is then an SDP formulation, and we call it the {\em primal problem} for the interactive instrument robustness. 
The theory of duality can then be applied, leading to \eqref{Eq: main dual}.
%
%
%
To see this, we can now compute its dual problem.
The Lagrangian is
\begin{align}
\mathcal{L}(\{\omega_a^{AB}\}_a,\{Y_a^B\}_a)&=\sum_a\tr\left(\kappa_a^B\right)-1\nonumber\\
&\quad-\sum_a \tr\left[\omega_a^{AB}\left(\frac{\mathds{1}^A}{d_A}\otimes\kappa_a^B-\mathcal{J}(\mcE_a)^{AB}\right)\right]\nonumber\\
&\quad-\sum_a \tr\left(Y_a^B\kappa_a^B\right)\nonumber\\
&=\sum_a \tr\left[\kappa_a^B \left(\mathds{1}^B-\frac{1}{d_A}\tr_A(\omega_a^{AB})-Y_a^B\right)\right]\nonumber\\
&\quad-1+\sum_a \tr\left(\omega_a^{AB}\mathcal{J}(\mcE_a)^{AB}\right).
\end{align}
We can then write the primal problem \eqref{Eq:primal APP} into
\begin{align}
R({\bm\mcE}^{\AtoB{A}{B}})= \min_{\kappa_a=\kappa_a^\dagger}\max_{\substack{\omega^{AB}_a\succeq0,\\Y^B_a\succeq0}}\mathcal{L}(\{\omega_a^{AB}\}_a,\{Y_a^B\}_a).
\end{align}
Switching the order of minimisation and maximisation, we obtain the {\em dual problem} as
\begin{align}
R({\bm\mcE}^{\AtoB{A}{B}})= \max_{\substack{\omega^{AB}_a\succeq0,\\Y^B_a\succeq0}}\min_{\kappa_a=\kappa_a^\dagger}\mathcal{L}(\{\omega_a^{AB}\}_a,\{Y_a^B\}_a).
\end{align}
\CY{This is} equal to $R({\bm\mcE}^{\AtoB{A}{B}})$ since the strong duality holds, as there exists a strictly feasible solution to the primal problem, e.g., take $\kappa_a^B=\mathds{1}^B\succ0$ $\forall\,a$, which gives $\mathcal{J}(\mcE_a)^{AB}\prec(\mathds{1}^A/d_A)\otimes\kappa_a^B$ $\forall\,a$ \CY{(known as {\em Slater's condition}~\cite{Slater2014}; see also, e.g., Refs.~\cite{Watrous_2018,Tavakoli2024RMP})}.
To solve the dual problem, note that $\min_{\kappa_a=\kappa_a^\dagger}\mathcal{L}(\{\omega_a^{AB}\}_a,\{Y_a^B\}_a)$ (which minimises over Hermitian $\kappa_a^B$'s) equals $-\infty$ unless $0=\mathds{1}^B-(1/d_A)\tr_A(\omega_a^{AB})-Y_a^B$, which implies $\tr_A(\omega_a^{AB})\preceq d_A\mathds{1}^B$ since $Y_a^B\succeq 0$. Thus, the dual problem becomes
\begin{equation}
\begin{split}\label{Eq: original dual}
R({\bm\mcE}^{\AtoB{A}{B}})=  \quad \max&\quad \sum_a\tr\left( \omega_a^{AB}\mathcal{J}(\mcE_a)^{AB}\right)-1\\
  \text{s.t.}& \quad \tr_A(\omega_a^{AB})\preceq d_A\mathds{1}^B,\\
  &\quad\omega^{AB}_a \succeq 0\quad\forall\,a.
\end{split}
\end{equation}
Next, we can replace the inequality constraint with an {\em equality}, namely,
\begin{equation}
\begin{split}\label{dual}
R({\bm\mcE})=  \quad \max&\quad \sum_a\tr\left( \omega_a^{AB}\mathcal{J}(\mcE_a)^{AB}\right)-1\\
  \text{s.t.}& \quad \tr_A(\omega_a^{AB})=d_A\mathds{1}^B,\\
  &\quad \omega^{AB}_a \succeq 0\quad\forall\,a.
\end{split}
\end{equation}
To justify this, note first that \eqref{dual}'s solution is also a solution for \eqref{Eq: original dual}. 
For the converse, consider an optimal solution $\{\widetilde{\omega}_{a}^{AB}\}_a$ to \eqref{Eq: original dual}. Then there exists $R_a^{B}\succeq0$ such that $\tr_A(\widetilde{\omega}_a^{AB})+R_a^{B}=d_A\mathds{1}^B$. 
The new solution $\{\widetilde{\omega}_a^{AB}+(\mathds{1}^A/d_A)\otimes R_a^B\}_a$ satisfies the constraints in both formulations, but it leads to \eqref{Eq: original dual}'s objective picking up the contribution $\sum_a \tr\left[\left((\mathds{1}^A/d_A)\otimes R_a^B\right)\mathcal{J}(\mcE_a)^{AB}\right]\geq 0$. Hence, the objective cannot decrease, which contradicts the optimality of \eqref{Eq: original dual}'s solution unless $\sum_a \tr\left[\left((\mathds{1}^A/d_A)\otimes R_a^B\right)\mathcal{J}(\mcE_a)^{AB}\right]=0$ is precisely zero. This then implies that the new solution $\{\widetilde{\omega}_a^{AB}+(\mathds{1}^A/d_A)\otimes R_a^B\}_a$ must also give the same, {\em optimal} objective value for \eqref{Eq: original dual} as the one achieved by $\{\widetilde{\omega}_a^{AB}\}_a$---the former formulation's optimality can indeed be achieve by a solution of the latter formulation. Hence, the two formulations are equivalent.

Finally, we obtain the desired dual problem as given in \eqref{Eq: main dual}
by rescaling $\omega_a^{AB}\rightarrow \omega_a^{AB}/d_A$.
As mentioned in the main text, this dual problem is useful because it lets us interpret $\omega_a^{AB}$ as the (un-normalised) Choi operator of a completely-positive {\em unital} linear map $\Gamma_a^{\AtoB{A}{B}}$ from $A$ to $B$ (i.e., it maps $\mathds{1}^A$ to $\mathds{1}^B$). 
More precisely, we can write (note the $d_A$ dependence)
\begin{align}
\omega^{AB}_a=({\rm id}^A\otimes\Gamma_a^{\AtoB{A}{B}})(d_A\Phi_+^{AA})\quad\forall\,a,
\end{align}
which ensures that $\mathds{1}^A$ is mapped to $\mathds{1}^B$.
Since $\hat{\Lambda}^{\AtoB{B}{A}}_a\coloneqq(\Gamma_a^{\AtoB{A}{B}})^\dagger$ is a channel (i.e., trace-preserving) from $B$ to $A$, this means that we can write $1+R({\bm\mcE}^{\AtoB{A}{B}})$ as the following maximisation:
\begin{equation}
\begin{split}\label{dual alrternative}
\max&\quad d_A^2\sum_a\tr\left[ ({\rm id}^A\otimes(\hat{\Lambda}^{\AtoB{B}{A}}_a)^\dagger)(\Phi_+^{AA})\mathcal{J}(\mcE_a)^{AB}\right]\\
  \text{s.t.} &\quad \text{each $\hat{\Lambda}^{\AtoB{B}{A}}_a$ is a channel from $B$ to $A$}.
\end{split}
\end{equation}
Taking the adjoint map can thus prove Result~\ref{Result:F_+}.\\

\section{Proof of Result~\ref{Result:F_ave R}
}\label{App:Proof-Result:F_ave R}

\noindent{\em Proof.}
Recall from the main text 
\CYSuppMat{Eq.~(6)}
that the average fidelity of an instrument ${\bm\mcE}^{\AtoB{A}{B}}=\{\mcE_a^{\AtoB{A}{B}}\}_a$ from $A$ to $B$ reads
\begin{equation}\label{Eq:F_ave def SM}
F_{\rm ave}({\bm\mcE}^{\AtoB{A}{B}})=\max_{\{\hat{\Lambda}_a^{\AtoB{B}{A}}\}_a}\int d\psi \bracket{\psi^A}{\sum_a (\hat{\Lambda}_a\circ \mcE_a)^A(\psi^A)}{\psi^A}.
\end{equation}
To compute this quantity, we use the inverse Choi map~\cite{Choi1975,Jamiolkowski1972}. 
More precisely, for the trace non-increasing linear map $\mcE_a^{\AtoB{A}{B}}$ from $A$ to $B$ and the unital linear map $\Gamma_a^{\AtoB{A}{B}}\coloneqq(\hat{\Lambda}_a^{\AtoB{B}{A}})^\dagger$ from $A$ to $B$ (since $\hat{\Lambda}_a^{\AtoB{B}{A}}$ is from $B$ to $A$), we can write, for every $\hat{\rho}^A$ in $A$,
\begin{align}
&\mcE_a^{\AtoB{A}{B}}(\hat{\rho}^A)=d_A\tr_A\left[\left((\hat{\rho}^A)^t\otimes\mathds{1}^B\right)\mathcal{J}(\mcE_a)^{AB}\right];\\
&\Gamma_a^{\AtoB{A}{B}}(\hat{\rho}^A)=d_A\tr_A\left[\left((\hat{\rho}^A)^t\otimes\mathds{1}^B\right)\mathcal{J}(\Gamma_a)^{AB}\right]
\end{align}
with the Choi operators~\cite{Choi1975,Jamiolkowski1972} as defined in \eqref{Eq:Choi operators}, namely,
\mbox{$
\mathcal{J}(\mcE_a)^{AB}\coloneqq({\rm id}^A\otimes\mcE_a^{\AtoB{A}{B}})(\Phi_+^{AA})
$}
and
\mbox{$
\mathcal{J}(\Gamma_a)^{AB}\coloneqq({\rm id}^A\otimes\Gamma_a^{\AtoB{A}{B}})(\Phi_+^{AA}),
$}
and $(\cdot)^t$ denotes the transpose map.
Then, we have (below, $A'$ denotes an identical copy of the system $A$; \CYtwo{also, recall that $\psi\coloneqq\proj{\psi}$})
\begin{widetext}
\begin{align}
\bracket{\psi^A}{(\hat{\Lambda}_a\circ\mcE_a)^A(\psi^A)}{\psi^A}&=\tr\left[\mcE_a^{\AtoB{A}{B}}(\psi^A)(\hat{\Lambda}_a^{\AtoB{B}{A}})^\dagger(\psi^A)\right]\eqqcolon\tr\left[\mcE_a^{\AtoB{A}{B}}(\psi^A)\Gamma_a^{\AtoB{A}{B}}(\psi^A)\right]\nonumber\\\nonumber
&=d_A^2\tr\Big(\tr_A\left[\left((\psi^A)^t\otimes\mathds{1}^B\right)\mathcal{J}(\mcE_a)^{AB}\right] \tr_{A'}\left[\left((\psi^{A'})^t\otimes\mathds{1}^B\right)\mathcal{J}(\Gamma_a)^{A'B}\right]\Big)\\
&=d_A^2\tr\Big(\left[(\psi^{A})^t\otimes(\psi^{A'})^t\otimes\mathds{1}^B\right](\mathds{1}^{A'}\otimes\mathcal{J}(\mcE_a)^{AB})(\mathds{1}^{A}\otimes\mathcal{J}(\Gamma_a)^{A'B})\Big).
\end{align}
This further implies
\begin{align}
F_{\rm ave}({\bm\mcE}^{\AtoB{A}{B}})=\max_{\{\Gamma_a^{\AtoB{A}{B}}\}_a}d_A^2\int d\psi  \sum_a 
\tr\Big(\left[(\psi^{A})^t\otimes(\psi^{A'})^t\otimes\mathds{1}^B\right](\mathds{1}^{A'}\otimes\mathcal{J}(\mcE_a)^{AB})(\mathds{1}^{A}\otimes\mathcal{J}(\Gamma_a)^{A'B})\Big),
\end{align}
which now maximises over unital maps $\Gamma_a^{\AtoB{A}{B}}$ from $A$ to $B$.
Since trace is invariant under partial transpose, we can apply the partial transpose to systems $AA'$ and obtain
\begin{equation}
F_{\rm ave}({\bm\mcE}^{\AtoB{A}{B}})=\max_{\{\Gamma_a^{\AtoB{A}{B}}\}_a}d_A^2\sum_a \tr\left(\left[\int d\psi(\psi^A \otimes \psi^{A'})\otimes\mathds{1}^B\right]
\left[\mathds{1}^{A'}\otimes(\mathcal{J}(\mcE_a)^{AB})^{t_A}\right]\left[\mathds{1}^{A}\otimes(\mathcal{J}(\Gamma_a)^{A'B})^{t_{A'}}\right]\right),
\end{equation}
where $(\cdot)^{t_A}$ and $(\cdot)^{t_{A'}}$ are the partial transpose in $A$ and $A'$, respectively.
The integral is a Haar-average over all two-copy states, which equals the normalised projector onto the symmetric subspace (e.g., set $n=2$ in Proposition 6 in Ref.~\cite{Harrow2013}):
\begin{equation}
\int d\psi\psi^A \otimes \psi^{A'}=\frac{2}{d_A(d_A+1)}\Pi_{\text{sym}}^{AA'},
\end{equation}
\CYtwo{where $\Pi_{\text{sym}}^{AA'}$ is the projector onto the symmetric subspace that satisfies} $\Pi_{\text{sym}}^{AA'}\left(\ket{\psi}^A\otimes\ket{\psi}^{A'}\right)=\ket{\psi}^A\otimes\ket{\psi}^{A'}$ for any $\ket{\psi}$, and $d_A(d_A+1)/2$ is the dimension of the symmetric subspace (see, e.g., Proposition 6 and Table 1 in Ref.~\cite{Harrow2013}). 
Now, recall that $\Pi_{\rm sym}^{AA'}$ can be expressed as a combination of $\mathds{1}^{AA'}$ and the swap operator ${\rm SWAP}^{AA'}:\rho^A\otimes\sigma^{A'}\mapsto\sigma^{A}\otimes\rho^{A'}$ $\forall\,\rho,\sigma$ (see, e.g., Proposition 7.1 with $n=2$ in Ref.~\cite{Watrous_2018}):
\begin{equation}
\Pi_{\text{sym}}^{AA'}=\frac{1}{2}\left( \mathds{1}^{AA'}+ {\rm SWAP}^{AA'}\right)=\frac{1}{2}\left( \mathds{1}^{AA'}+ d_A(\Phi^{AA'}_+)^{t_A}\right),
\end{equation}
where note that ${\rm SWAP}^{AA'}=d_A (\Phi^{AA'}_+)^{t_A}$, which follows directly from definition.
Returning to the fidelity expression, we now have
\begin{equation}
F_{\rm ave}({\bm\mcE}^{\AtoB{A}{B}})=\max_{\{\Gamma_a^{\AtoB{A}{B}}\}_a}\frac{d_A}{d_A+1}  \sum_a \tr\Bigg(\left(\left[\mathds{1}^{AA'}+ d_A(\Phi^{AA'}_+)^{t_A}\right]\otimes\mathds{1}^B\right) \left[\mathds{1}^{A'}\otimes(\mathcal{J}(\mcE_a)^{AB})^{t_A}\right]\left[\mathds{1}^{A}\otimes(\mathcal{J}(\Gamma_a)^{A'B})^{t_{A'}}\right]\Bigg).
\end{equation}
Again, use the fact that trace is invariant under partial transpose on $A$, we obtain
\begin{align}\nonumber
F_{\rm ave}({\bm\mcE}^{\AtoB{A}{B}})&=\max_{\{\Gamma_a^{\AtoB{A}{B}}\}_a}\frac{d_A}{d_A+1}  \sum_a \tr\Bigg(\left[\left(\mathds{1}^{AA'}+ d_A\Phi^{AA'}_+\right)\otimes\mathds{1}^B\right] \left[\mathds{1}^{A'}\otimes\mathcal{J}(\mcE_a)^{AB}\right]\left[\mathds{1}^{A}\otimes(\mathcal{J}(\Gamma_a)^{A'B})^{t_{A'}}\right]\Bigg)\\ \nonumber
& =\max_{\{\Gamma_a^{\AtoB{A}{B}}\}_a}\bigg\{ \frac{d_A}{d_A+1}\sum_a \tr\Big(\left[\mathds{1}^{A'}\otimes\mathcal{J}(\mcE_a)^{AB}\right]\left[\mathds{1}^{A}\otimes(\mathcal{J}(\Gamma_a)^{A'B})^{t_{A'}}\right]\Big)\\
&\qquad\qquad + \frac{d_A^2}{d_A+1}  \sum_a \tr\Big(\left(\Phi^{AA'}_+\otimes\mathds{1}^B\right)\left[\mathds{1}^{A'}\otimes\mathcal{J}(\mcE_a)^{AB}\right]\left[\mathds{1}^{A}\otimes(\mathcal{J}(\Gamma_a)^{A'B})^{t_{A'}}\right]\Big)\bigg\}.
\end{align}
For the term in the second line, we apply a partial transpose on $A'$, which preserves the trace. 
For the term in the last line, we apply the formula $(R^A\otimes\mathds{1}^{A'})\ket{\Phi_+^{AA'}}=(\mathds{1}^{A}\otimes (R^{A'})^t)\ket{\Phi_+^{AA'}}$ that holds for every normal operator $R$.
Then, we obtain
\begin{align}\nonumber\label{stepcalc}
F_{\rm ave}({\bm\mcE}^{\AtoB{A}{B}})&=\max_{\{\Gamma_a^{\AtoB{A}{B}}\}_a}
\bigg\{  \frac{d_A}{d_A+1}\sum_a \tr\Big(\left[\mathds{1}^{A'}\otimes\mathcal{J}(\mcE_a)^{AB}\right]\left[\mathds{1}^{A}\otimes\mathcal{J}(\Gamma_a)^{A'B}\right]\Big)\\
&\qquad\qquad + \frac{d_A^2}{d_A+1}  \sum_a \tr\Big(\left(\Phi^{AA'}_+\otimes\mathds{1}^B\right)\left[\mathds{1}^{A}\otimes(\mathcal{J}(\mcE_a)^{A'B})^{t_{A'}}\right]\left[\mathds{1}^{A}\otimes(\mathcal{J}(\Gamma_a)^{A'B})^{t_{A'}}\right]\Big)\bigg\}.
\end{align}
Recall that $\Gamma^{\AtoB{A}{B}}_a$ is a unital linear map from $A$ to $B$, meaning that it maps $\mathds{1}^A$ to $\mathds{1}^B$.
This means that 
\begin{align}
\tr_{A}\left(\mathcal{J}(\Gamma_a)^{AB}\right)=\tr_A\left[({\rm id}^A\otimes\Gamma_a^{\AtoB{A}{B}})(\Phi_+^{AA})\right]=\frac{\Gamma_a^{\AtoB{A}{B}}\left(\mathds{1}^A\right)}{d_A}=\frac{1}{d_A}\mathds{1}_B.
\end{align}
Hence, by tracing out $A'$ in the first term inside the maximisation in \eqref{stepcalc}, it becomes (note that $d_{A'}=d_A$)
\begin{align}
&\frac{d_A}{d_A+1}\sum_a \tr\Big(\left[\mathds{1}^{A'}\otimes\mathcal{J}(\mcE_a)^{AB}\right]\left[\mathds{1}^{A}\otimes\mathcal{J}(\Gamma_a)^{A'B}\right]\Big)=\frac{1}{d_A+1}\sum_a \tr\Big(\mathcal{J}(\mcE_a)^{AB}\Big)=\frac{1}{d_A+1},
\end{align}
where note that that $\mcE_a^{\AtoB{A}{B}}$'s form an instrument and hence $\sum_a\mcE_a$ is a channel, thereby $\sum_a \tr\Big(\mathcal{J}(\mcE_a)^{AB}\Big)=\tr\Big(\mathcal{J}(\sum_a\mcE_a)^{AB}\Big)=1$.
Next, for the second term inside the maximisation in \eqref{stepcalc}, tracing out $A$ and then applying a partial transpose to $A'$ gives
\begin{equation}
\frac{d_A^2}{d_A+1}  \sum_a \tr\Big(\left(\Phi^{AA'}_+\otimes\mathds{1}^B\right)\left[\mathds{1}^{A}\otimes(\mathcal{J}(\mcE_a)^{A'B})^{t_{A'}}\right]\left[\mathds{1}^{A}\otimes(\mathcal{J}(\Gamma_a)^{A'B})^{t_{A'}}\right]\Big)=\frac{d_A}{d_A+1}  \sum_a \tr\left(\mathcal{J}(\mcE_a)^{A'B}\mathcal{J}(\Gamma_a)^{A'B}\right).
\end{equation} 
Putting everything together and use \eqref{dual alrternative} (as $\Gamma_a^{\AtoB{A}{B}}$'s are unital linear maps), we obtain
\begin{equation}\label{Eq:result Fave}
F_{\rm ave}({\bm\mcE}^{\AtoB{A}{B}})=\frac{1}{d_A+1}\left(1+d_A \max_{\{\Gamma_a^{\AtoB{A}{B}}\}_a}\sum_a \tr\left[\mathcal{J}(\mcE_a)^{A'B}\mathcal{J}(\Gamma_a)^{A'B}\right]\right)=\frac{1}{d_A+1}\left(1+\frac{R\left({\bm\mcE}^{\AtoB{A}{B}}\right)+1}{d_A}\right)=\frac{1}{d_A}\left(1+\frac{R\left({\bm\mcE}^{\AtoB{A}{B}}\right)}{d_A+1}\right).
\end{equation}
This then concludes the proof.
\hfill$\square$
\newpage
\end{widetext}

\section{Proof of 
Result~\ref{Result: three operational meanings}
}\label{App:Proof-Result: three operational meanings}

First, we need to establish the following lemma \CY{(recall that ${\bf NI}$ is the set of all non-interactive isntruments)}:\\

\begin{lemma}\label{small lemma}
$\tr_A(\omega_a^{AB})\preceq d_A\mathds{1}^B$ $\forall\,a$ if and only if 
\begin{align}\label{Eq:small lemma 001}
\max_{{\bm{\mathcal{L}}}^{\AtoB{A}{B}}\in\CY{\bf NI}}\sum_a\tr[\omega_a^{AB}\left({\rm id}^A\otimes{\mathcal{L}}_a^{\AtoB{A}{B}}\right)(\Phi_+^{AA})]\le1.
\end{align}
\end{lemma}
{\em Proof of Lemma~\ref{small lemma}.}
We first prove that $\tr_A(\omega_a^{AB})\preceq d_A\mathds{1}^B$ $\forall\,a$ if and only if 
\begin{align}\label{Eq:small lemma}
\sum_a\tr[\omega_a^{AB}\left(\frac{\mathds{1}^A}{d_A}\otimes\tau_a^B\right)]\le1\quad\text{{$\forall$} $\tau_a^B\succeq0$ with $\sum_a\tr(\tau_a^B)=1$.}
\end{align} 
If $\tr_A(\omega_a^{AB})\preceq d_A\mathds{1}^B$ $\forall\,a$, then we have, for every $\tau_a^B\succeq0$ with $\sum_a\tr(\tau_a^B)=1$,
\begin{align}\label{Eq:small lemma 002}
\sum_a\tr[\omega_a^{AB}\left(\frac{\mathds{1}^A}{d_A}\otimes\tau_a^B\right)]&=\sum_a\frac{\tr\left[\tr_A(\omega^{AB}_a)\tau_a^B\right]}{d_A}\nonumber\\
&\le\sum_a\tr\left(\mathds{1}^B\tau_a^B\right)=1.
\end{align}
On the other hand, if \eqref{Eq:small lemma} holds, then we have that, for every {\em fixed} $a$, by setting $\tau_b^B=\delta_{a,b}\hat{\rho}^B$ with some state $\hat{\rho}^B$,
\begin{align}
\tr\left[\tr_A(\omega^{AB}_a)\hat{\rho}^B\right]\le d_A\quad\forall\,\text{state $\hat{\rho}^B$}.
\end{align}
This means that
$
\tr[\left(d_A\mathds{1}^B-\tr_A(\omega^{AB}_a)\right)\hat{\rho}^B]\ge0
$
for every state $\hat{\rho}^B$,
which is equivalent to $d_A\mathds{1}^B-\tr_A(\omega^{AB}_a)\succeq0$ (note that the argument holds for every $a$).

Now, \CYtwo{note that an instrument ${\bm{\mathcal{L}}}^{\AtoB{A}{B}}$ from $A$ to $B$ is non-interactive if and only if} its Choi representations read 
\CYSuppMat{[from Eq.~(15) in the main text]}
$\left({\rm id}^A\otimes\mathcal{L}_a^{\AtoB{A}{B}}\right)(\Phi_+^{AA})=\CYtwo{(\mathds{1}^A/d_A)\otimes\tau^B_a}$ for some \CYtwo{$\tau_a^B\succeq0$ with $\sum_a{\rm tr}(\tau_a^B)=1$.}
\CYtwo{Hence,} \eqref{Eq:small lemma} holds if and only if 
\begin{align}
\sum_a\tr[\omega_a^{AB}\left({\rm id}^A\otimes{\mathcal{L}}_a^{\AtoB{A}{B}}\right)(\Phi_+^{AA})]\le1\quad\forall\,{\bm{\mathcal{L}}}^{\AtoB{A}{B}}\in\CY{\bf NI},
\end{align}
which is true if and only if \eqref{Eq:small lemma 001} holds.
This concludes the proof of the lemma.
\hfill$\square$\\

\CYtwo{Before the proof of 
Result~\ref{Result: three operational meanings}, let us define a notation.}
For a given set $\{\omega_a^{AB}\}_a$, let us write 
\begin{align}
\beta(\{\omega_a^{AB}\}_a)\coloneqq\max_{{\bm{\mathcal{L}}}^{\AtoB{A}{B}}\in\CY{\bf NI}}\sum_a\tr[\omega_a^{AB}\left({\rm id}^A\otimes{\mathcal{L}}_a^{\AtoB{A}{B}}\right)(\Phi_+^{AA})].
\end{align}
Then, from \eqref{Eq:small lemma} and \eqref{Eq:small lemma 002}, one can observe that
\begin{align}
\beta(\{\omega_a^{AB}\}_a)&=\max_{\substack{\tau_a^B\succeq0,\\\sum_a\tr(\tau_a^B)=1}}\sum_a\frac{\tr\left[\tr_A(\omega^{AB}_a)\tau_a^B\right]}{d_A}\nonumber\\
&\ge\frac{1}{d_A}\norm{\tr_A(\omega^{AB}_b)}_\infty\quad\forall\,b,
\end{align}
where the last inequality follows by considering a sub-optimal maximisation over $\tau_a^B=\delta_{a,b}\hat{\rho}^B$ with state $\hat{\rho}^B$, which will then return $\norm{\tr_A(\omega^{AB}_b)}_\infty/d_A$.
Since $\norm{\tr_A(\omega^{AB}_b)}_\infty=0$ if and only if $\tr_A(\omega^{AB}_b)=0$, which can be true if and only if $\omega^{AB}_b=0$, we thus conclude that
\begin{align}\label{Eq:small corollary}
\beta(\{\omega_a^{AB}\}_a)=0\quad\text{if and only if}\quad\omega_a^{AB}=0\quad\forall\,a.
\end{align}

Now, we can proceed to prove Result~\ref{Result: three operational meanings}.
\\

\noindent{\em Proof of Result~\ref{Result: three operational meanings}.}
Recall from \eqref{Eq:Choi operators} that, for an instrument ${\bm\mcE}^{\AtoB{A}{B}}=\{\mcE_a^{\AtoB{A}{B}}\}_a$ from $A$ to $B$, we write its Choi operator~\cite{Choi1975,Jamiolkowski1972} as $\mathcal{J}(\mcE_a)^{AB}\coloneqq\left({\rm id}^A\otimes\mcE_a^{\AtoB{A}{B}}\right)(\Phi_+^{AA})$.
Then, from \eqref{Eq: original dual}, we have
\begin{equation}
\begin{split}
1+R({\bm\mcE}^{\AtoB{A}{B}})=  \quad \max&\quad \sum_a\tr\left( \omega_a^{AB}\mathcal{J}(\mcE_a)^{AB}\right)\\
 \text{s.t.} & \quad \tr_A(\omega_a^{AB})\preceq d_A\mathds{1}^B,\\
&\quad \omega^{AB}_a \succeq 0\quad\forall\,a.
\end{split}
\end{equation}
Using Lemma~\ref{small lemma}, the above optimisation can be rewritten into
\begin{equation}
\begin{split}\label{Eq: dual form 0001-2}
1+R({\bm\mcE}^{\AtoB{A}{B}})=  \quad \max&\quad \sum_a\tr\left( \omega_a^{AB}\mathcal{J}(\mcE_a)^{AB}\right)\\
  \text{s.t.}& \quad \beta(\{\omega_a^{AB}\}_a)\le1,\\
&\quad \omega^{AB}_a \succeq 0\quad\forall\,a.
\end{split}
\end{equation}
Now, from \eqref{Eq:small corollary}, we note that $\beta(\{\omega_a^{AB}\}_a)=0$ if and only if $\omega_a^{AB}=0$ $\forall\,a$.
On the other hand, $1+R({\bm\mcE}^{\AtoB{A}{B}})\ge1$, meaning that the optimality must be achieved by some non-vanishing $\omega_a^{AB}$'s.
Hence, without loss of generality, we can assume $\beta(\{\omega_a^{AB}\}_a)>0$, meaning that $1\le\beta(\{\omega_a^{AB}\}_a)^{-1}$.
Namely, we have
\begin{align}\label{Eq: upper bound 001-2}
1+R({\bm\mcE}^{\AtoB{A}{B}})\le \max_{\omega_a^{AB} \succeq 0}\frac{\sum_a\tr\left( \omega_a^{AB}\mathcal{J}(\mcE_a)^{AB}\right)}{\beta(\{\omega_a^{AB}\}_a)}.
\end{align}
To see how this upper bound is saturated, suppose $\omega_a^{AB}$'s are feasible solutions to the above maximisation.
Then, define
\begin{align}
\widetilde{\omega}^{AB}_a\coloneqq\frac{\omega_a^{AB}}{\beta(\{\omega_b^{AB}\}_b)}\quad\forall\,a.
\end{align}
Then one can directly check that $\beta(\{\widetilde\omega_a^{AB}\}_a)=1$, and hence $\widetilde\omega^{AB}_a$'s form a feasible solution to \eqref{Eq: dual form 0001-2}.
This thus means that
\begin{align}
1+R({\bm\mcE}^{\AtoB{A}{B}})\ge\sum_a\tr\left(\widetilde\omega_a^{AB}\mathcal{J}(\mcE_a)^{AB}\right) =\frac{\sum_a\tr\left( \omega_a^{AB}\mathcal{J}(\mcE_a)^{AB}\right)}{\beta(\{\omega_a^{AB}\}_a)}.
\end{align}
Since this argument holds for {\em all} feasible solutions $\omega_a^{AB}$'s to the maximisation in \eqref{Eq: upper bound 001-2}, we conclude that
\begin{align}
1+R({\bm\mcE}^{\AtoB{A}{B}})=\max_{\omega_a^{AB}\succeq 0}\frac{\sum_a\tr\left( \omega_a^{AB}\mathcal{J}(\mcE_a)^{AB}\right)}{\beta(\{\omega_a^{AB}\}_a)}.
\end{align}
Dividing both the numerator and denominator by $\norm{\sum_a\omega_a^{AB}}_\infty$ and \CYtwo{defining ${\bm Q}^{AB}=\{Q_a^{AB}\}_a$ with $a=1,...,|\mathcal{A}|,\emptyset$ as }
\begin{align}
&Q_a^{AB}\coloneqq\frac{\omega_a^{AB}}{\norm{\sum_b\omega_b^{AB}}_\infty}\quad\forall\,a=1,...,|\mathcal{A}|;\\
&\CYtwo{Q_\emptyset^{AB}\coloneqq\mathds1^{AB}-\sum_{a=1}^{|\mathcal{A}|}Q_a^{AB}},
\end{align}
which satisfies $Q_a^{AB}\succeq0$ \CYtwo{$\forall\,a=1,...,|\mathcal{A}|$} and $\sum_{a=1}^{|\mathcal{A}|}Q_a^{AB}\preceq\mathds1^{AB}$, \CYtwo{meaning that $0\preceq Q_\emptyset^{AB}\preceq\mathds1^{AB}$ and hence ${\bm Q}^{AB}$ is a valid POVM}. 
\CYtwo{Using 
\CYSuppMat{Eq.~(9)}
in the main text to write $P_{\rm succ}({\bm\mcE}^{\AtoB{A}{B}},{\bm Q}^{AB})=\sum_{a=1}^{|\mathcal{A}|}\tr\left[Q_a^{AB}\mathcal{J}(\mcE_a)^{AB}\right]$,} we obtain
\begin{align}
1+R({\bm\mcE}^{\AtoB{A}{B}})&=\max_{{\bm Q}^{AB}}\frac{\sum_{a=1}^{|\mathcal{A}|}\tr\left[Q_a^{AB}\mathcal{J}(\mcE_a)^{AB}\right]}{\max_{{\bm{\mathcal{L}}^{\AtoB{A}{B}}}\in\CY{\bf NI}}\sum_{a=1}^{|\mathcal{A}|}\tr\left[Q_a^{AB}\mathcal{J}(\mathcal{L}_a)^{AB}\right]}\nonumber\\
&=\max_{{\bm Q}^{AB}}\frac{P_{\rm succ}({\bm\mcE}^{\AtoB{A}{B}},{\bm Q}^{AB})}{\max_{{\bm{\mathcal{L}}}^{\AtoB{A}{B}}\in\CY{\bf NI}}P_{\rm succ}({\bm{\mathcal{L}}}^{\AtoB{A}{B}},{\bm Q}^{AB})}.
\end{align} 
This thus completes the proof.
\hfill$\square$\\

As a remark, recall from \eqref{Eq: trivial instrument Choi} that a \CY{non-interactive} instrument ${\bm{\mathcal{L}}}^{\AtoB{A}{B}}$ has its Choi representations as $\left({\rm id}^A\otimes\mathcal{L}_a^{\AtoB{A}{B}}\right)(\Phi_+^{AA})=r(a)(\mathds{1}^A/d_A)\otimes\hat{\tau}^B_a$ for some states $\hat{\tau}_a^B$ and a probability $r(a)$. 
We then note that
\begin{align}
&\max_{{\bm{\mathcal{L}}}^{\AtoB{A}{B}}\in\CY{\bf NI}}P_{\rm succ}({\bm{\mathcal{L}}}^{\AtoB{A}{B}},{\bm Q}^{AB})
=\max_{\substack{\{r(a)\}_a,\\\{\hat{\tau}_a^B\}_a}}\sum_{a=1}^{|\mathcal{A}|}r(a)\tr\left[\left(\frac{\tr_A(Q_a^{AB})}{d_A}\right)\hat{\tau}^B_a\right]\nonumber\\
&=\max_{\{r(a)\}_a}\sum_{a=1}^{|\mathcal{A}|}r(a)\norm{\frac{\tr_A(Q_a^{AB})}{d_A}}_\infty
=\max_{a\in\{1,...,|\mathcal{A}|\}}\frac{\norm{\tr_A\left(Q_a^{AB}\right)}_\infty}{d_A}.
\end{align}

\section{Proof of 
Result~\ref{Result:conversion iff conditions}
}\label{App:Proof-Result:conversion iff conditions}

Before the proof, \CYtwo{let us detail some necessary tools.
Recall from \eqref{free op} that the set of allowed operation can be written as ${\bf AO}={\rm conv}(\text{\bf pre-}{\bf AO})$ (where ``${\rm conv}$'' denotes the convex hull), where $\text{\bf pre-}{\bf AO}$ is the set of all transformations of the form:
\begin{align}\label{Eq:pre-AO}
\Pi({\bm\mcE})_a=\sum_kP(a|k)(\hat{\mathcal{D}}_k\circ\mcE_k\circ\hat{\mathcal{P}})(\cdot)\quad\forall\,a,{\bm\mcE},
\end{align}
where $\hat{\mathcal{P}}$ and $\hat{\mathcal{D}}_k$'s are all channels, and $\{P(a|k)\}_{a,k}$ are conditional probability distributions.
Similarly, for fixed input and output dimensions, the set ${\bf AO}^{\AtoB{A}{B}}={\rm conv}(\text{\bf pre-}{\bf AO}^{\AtoB{A}{B}})$, where $\text{\bf pre-}{\bf AO}^{\AtoB{A}{B}}$ is the set of transformations of the form \eqref{Eq:pre-AO} but maps any ${\bm\mcE}^{\AtoB{A}{B}}$ again to some instrument from $A$ to $B$: 
\begin{align}\label{Eq:pre-AO def}
\Pi({\bm\mcE}^{\AtoB{A}{B}})_a=\sum_kP(a|k)(\hat{\mathcal{D}}_k^B\circ\mcE_k^{\AtoB{A}{B}}\circ\hat{\mathcal{P}}^A)(\cdot)\quad\forall\,a,{\bm\mcE}^{\AtoB{A}{B}}.
\end{align}
It turns out that ${\bf AO}^{\AtoB{A}{B}}$ is compact.}
To formally make sense of this, we need to specify the topology.
Let us consider the following distance between two supermaps $\Pi,\Pi'$ [suppose their output instruments map from $A$ to $B$; also, we again use \CYtwo{``$\Pi({\bm\mcE})_a$''} to denote the $a$-th element of the instrument $\Pi({\bm\mcE})$]:
\begin{align}\label{Eq:metric}
&D(\Pi,\Pi')\coloneqq\nonumber\\
&\quad\max_{{\bm\mcE}}\sum_a\norm{\left[{\rm id}^A\otimes\left(\CYtwo{\Pi({\bm\mcE})_a-\Pi'({\bm\mcE})_a}\right)^{\AtoB{A}{B}}\right](\Phi_+^{AA})}_\infty,
\end{align}
which maximises over all \CYtwo{valid} input instruments.
Due to Choi representation~\cite{Choi1975,Jamiolkowski1972}, one can check that $D$ is a {\em metric}, namely, for every supermaps $\Pi,\Pi',\Pi''$, we have that:
(1) $D(\Pi,\Pi')\ge0$ and equality holds if and only if $\Pi=\Pi'$.
(2) $D(\Pi,\Pi')=D(\Pi',\Pi)$.
(3) (triangle inequality) $D(\Pi,\Pi'')\le D(\Pi,\Pi')+D(\Pi',\Pi'')$.
This metric can induce a metric topology.
Then, we have:\\

\begin{lemma}\label{lemma:compact}
Both $\text{\bf pre-}{\bf AO}^{\AtoB{A}{B}}$ and ${\bf AO}^{\AtoB{A}{B}}$ are compact in the metric topology induced by $D$.
\end{lemma}

\noindent{\em Proof of Lemma~\ref{lemma:compact}.}
Since the convex hull of a compact set is again compact, it suffices to show $\text{\bf pre-}{\bf AO}^{\AtoB{A}{B}}$'s compactness.
Consider $a=1,...,|\mathcal{A}|<\infty\;\;\&\;\;k=1,...,|\mathcal{K}|<\infty$.
We also write the inverse Choi map as~\cite{Choi1975,Jamiolkowski1972}
\begin{align}
\mathcal{J}^{-1}(\hat\rho^{AB})\coloneqq d_A\tr_A\left[\left(((\cdot)^A)^t\otimes\mathds{1}^B\right)\hat\rho^{AB}\right],
\end{align}
where, as one can directly check,
\begin{align}\label{Eq:Choi Choi}
\left({\rm id}^A\otimes\mathcal{J}^{-1}(\hat\rho^{AB})\right)(\Phi_+^{AA})=\hat\rho^{AB}\quad\forall\,\hat\rho^{AB}:{\rm state}.
\end{align}
Let ${\rm Choi}(X)$ be the set of all Choi states for channels acting on a finite-dimensional system $X$ and ${\rm Prob}_{\mathcal{A}|\mathcal{K}}\subset[0,1]^{|\mathcal{A}||\mathcal{K}|}$ be the set of all conditional probability distributions $\{P(a|k)\}_{a,k}$.
In our current setting, these sets are all compact in the ordinary Euclidean topology (as they are bounded and closed sets of finite-dimensional matrices).
Hence, their Cartesian product, denoted as ${\rm Domain}^{A\to B}_{\mathcal{A},\mathcal{K}}\coloneqq{\rm Prob}_{\mathcal{A}|\mathcal{K}}\times{\rm Choi}(B)^{|\mathcal{K}|}\times{\rm Choi}(A)$, is again a compact set in the product topology, which is again the ordinary Euclidean topology.
Now, define a function $f$ that maps from ${\rm Domain}^{A\to B}_{\mathcal{A},\mathcal{K}}$ (equipped with the ordinary Euclidean topology) to the set of supermaps [equipped with the metric topology induced by $D$ defined in \eqref{Eq:metric}] as follows:
\begin{align}
f(\{P(a|k)\},\{\hat{\rho}_{k}\},\hat{\sigma})\coloneqq\Pi,
\end{align}
where $\Pi$ is a supermap that takes the form
\begin{align}
\Pi({\bm\mcE})_a\coloneqq\sum_{k}P(a|k)\left[\mathcal{J}^{-1}(\hat{\rho}_{k})\circ\mcE_k\circ\mathcal{J}^{-1}(\hat{\sigma})\right](\cdot)\quad\forall\,a.
\end{align}
Then, by comparing with \eqref{Eq:pre-AO def}, we have that
\begin{align}\label{Eq: by construction}
f\left({\rm Domain}^{A\to B}_{\mathcal{A},\mathcal{K}}\right)=\text{\bf pre-}{\bf AO}^{\AtoB{A}{B}}.
\end{align}
Since ${\rm Domain}^{A\to B}_{\mathcal{A},\mathcal{K}}$ is compact in the ordinary Euclidean topology, it remains to show that $f$ is continuous in $D$ to conclude the compactness of $\text{\bf pre-}{\bf AO}^{\AtoB{A}{B}}$ in $D$'s metric topology.
To this end, consider two given tuples $(\{P(a|k)\},\{\hat{\rho}_{k}\},\hat{\sigma})$ and $(\{P'(a|k)\},\{\hat{\rho}'_{k}\},\hat{\sigma}')$. We estimate the distance between the supermaps \mbox{$\Pi=f(\{P(a|k)\},\{\hat{\rho}_{k}\},\hat{\sigma})$} and \mbox{$\Pi'=f(\{P'(a|k)\},\{\hat{\rho}'_{k}\},\hat{\sigma}')$} as follows (note that the maximisation for $D$ is now taken over instruments ${\bm\mcE}^{\AtoB{A}{B}}$ from $A$ to $B$)
\begin{widetext}
\begin{align}\label{Eq:estimate part one}
D(\Pi,\Pi')&\coloneqq\max_{{\bm\mcE}^{\AtoB{A}{B}}}\sum_a\norm{\left[{\rm id}^A\otimes\left(\Pi(\mcE_a)-\Pi'(\mcE_a)\right)^{\AtoB{A}{B}}\right](\Phi_+^{AA})}_\infty\nonumber\\
&=\max_{{\bm\mcE}^{\AtoB{A}{B}}}\sum_a\norm{\sum_{k}\left[{\rm id}^A\otimes\Bigg(P(a|k)\left[\mathcal{J}^{-1}(\hat{\rho}_{k})\circ\mcE_k\circ\mathcal{J}^{-1}(\hat{\sigma})\right]-P'(a|k)\left[\mathcal{J}^{-1}(\hat{\rho}'_{k})\circ\mcE_k\circ\mathcal{J}^{-1}(\hat{\sigma}')\right]\Bigg)^{\AtoB{A}{B}}\right](\Phi_+^{AA})}_\infty\nonumber\\
&\le\max_{{\bm\mcE}^{\AtoB{A}{B}}}\sum_{a,k}\Bigg\{\big|P(a|k)-P'(a|k)\big|\norm{\left[{\rm id}^A\otimes\left[\mathcal{J}^{-1}(\hat{\rho}_{k})\circ\mcE_k\circ\mathcal{J}^{-1}(\hat{\sigma})\right]^{\AtoB{A}{B}}\right](\Phi_+^{AA})}_\infty\nonumber\\
&\quad\quad\quad\quad\quad\quad+P'(a|k)\norm{\left[{\rm id}^A\otimes\left[\big(\mathcal{J}^{-1}(\hat{\rho}_{k})-\mathcal{J}^{-1}(\hat{\rho}'_{k})\big)\circ\mcE_k\circ\mathcal{J}^{-1}(\hat{\sigma})\right]^{\AtoB{A}{B}}\right](\Phi_+^{AA})}_\infty\nonumber\\
&\quad\quad\quad\quad\quad\quad+P'(a|k)\norm{\left[{\rm id}^A\otimes\left[\mathcal{J}^{-1}(\hat{\rho}'_{k})\circ\mcE_k\circ\Big(\mathcal{J}^{-1}(\hat{\sigma})-\mathcal{J}^{-1}(\hat{\sigma}')\Big)\right]^{\AtoB{A}{B}}\right](\Phi_+^{AA})}_\infty\Bigg\}\nonumber\\
&\le\max_{{\bm\mcE}^{\AtoB{A}{B}}}\Bigg\{\sum_{a,k}\big|P(a|k)-P'(a|k)\big|+\sum_{k}\norm{\left[{\rm id}^A\otimes\left[\big(\mathcal{J}^{-1}(\hat{\rho}_{k})-\mathcal{J}^{-1}(\hat{\rho}'_{k})\big)\circ\mcE_k\circ\mathcal{J}^{-1}(\hat{\sigma})\right]^{\AtoB{A}{B}}\right](\Phi_+^{AA})}_\infty\nonumber\\
&\quad\quad\quad\quad\quad\quad+\sum_{k}\norm{\left[{\rm id}^A\otimes\left[\mathcal{J}^{-1}(\hat{\rho}'_{k})\circ\mcE_k\circ\Big(\mathcal{J}^{-1}(\hat{\sigma})-\mathcal{J}^{-1}(\hat{\sigma}')\Big)\right]^{\AtoB{A}{B}}\right](\Phi_+^{AA})}_\infty\Bigg\},
\end{align}
\end{widetext}
where we have used the triangle inequality and $\sum_aP(a|k)=1=\sum_aP'(a|k)$ $\forall\,k$.
Now, we note that, for any linear map $\Delta$ and completely-positive trace-non-increasing linear map with operator sum $\mathcal{L}(\cdot)=\sum_i L_i(\cdot)L_i^\dagger$, we have (we drop the system dependence here for simplicity)
\begin{align}
&\norm{[{\rm id}\otimes(\Delta\circ\mathcal{L})](\Phi_+)}_\infty=\norm{(\mathcal{L}^t\otimes\Delta)(\Phi_+)}_\infty\nonumber\\
&\coloneqq\max_{\hat{\eta}}\left|{\rm tr}\left[(\mathcal{L}^{t,\dagger}\otimes{\rm id})(\hat{\eta})({\rm id}\otimes\Delta)(\Phi_+)\right]\right|\nonumber\\
&\le\norm{({\rm id}\otimes\Delta)(\Phi_+)}_\infty,
\end{align} 
where $\mathcal{L}^t(\cdot)\coloneqq\sum_i L_i^t(\cdot)L_i^{t,\dagger}$ and hence $\mathcal{L}^{t,\dagger}$ is again trace-non-increasing, and 
\begin{align}
&\norm{[{\rm id}\otimes(\mathcal{L}\circ\Delta)](\Phi_+)}_\infty
\nonumber\\
&\coloneqq\max_{\hat{\eta}}\Big|{\rm tr}\left[({\rm id}\otimes\mathcal{L}^\dagger)(\hat{\eta})({\rm id}\otimes\Delta)(\Phi_+)\right]\Big|
\nonumber\\
&\le\max_{\hat{\eta}}\norm{({\rm id}\otimes\mathcal{L}^\dagger)(\hat{\eta})({\rm id}\otimes\Delta)(\Phi_+)}_1\nonumber\\
&\le\max_{\hat{\eta}}\norm{({\rm id}\otimes\mathcal{L}^\dagger)(\hat{\eta})}_\infty\norm{({\rm id}\otimes\Delta)(\Phi_+)}_1\nonumber\\
&=\max_{\hat{\eta},\hat{\eta}'}\Big|{\rm tr}\left[({\rm id}\otimes\mathcal{L})(\hat{\eta}')\hat{\eta}\right]\Big|\norm{({\rm id}\otimes\Delta)(\Phi_+)}_1\nonumber\\
&\le\norm{({\rm id}\otimes\Delta)(\Phi_+)}_1,
\end{align}
where we have used $|{\rm tr}(P)|\le{\rm tr}|P|\eqqcolon\norm{P}_1$ and the H\"older's inequality $\norm{PQ}_1\le\norm{P}_\infty\norm{Q}_1$ for every (linear bounded) operators $P,Q$.
Consequently, \eqref{Eq:estimate part one} can be further upper bounded as
\begin{align}
D(\Pi,\Pi')&\le\sum_{a,k}\big|P(a|k)-P'(a|k)\big|\nonumber\\
&\quad+\sum_{k}\norm{\left[{\rm id}^B\otimes\big(\mathcal{J}^{-1}(\hat{\rho}_{k})-\mathcal{J}^{-1}(\hat{\rho}'_{k})\big)^{B}\right](\Phi_+^{BB})}_\infty\nonumber\\
&\quad+|\mathcal{K}|\norm{\left[{\rm id}^A\otimes\big(\mathcal{J}^{-1}(\hat{\sigma})-\mathcal{J}^{-1}(\hat{\sigma}')\big)^{A}\right](\Phi_+^{AA})}_1\nonumber\\
&\le\sum_{a,k}\big|P(a|k)-P'(a|k)\big|\nonumber\\
&\quad+\sum_{k}\norm{\hat{\rho}_{k}-\hat{\rho}'_{k}}_1+|\mathcal{K}|\norm{\hat{\sigma}-\hat{\sigma}'}_1,
\end{align}
where \eqref{Eq:Choi Choi} has been used.
Hence, the function $f$ from ${\rm Domain}^{A\to B}_{\mathcal{A},\mathcal{K}}$ (equipped with the ordinary Euclidean topology) to $f\left({\rm Domain}^{A\to B}_{\mathcal{A},\mathcal{K}}\right)$ (equipped with topology induced by $D$) is indeed continuous, which concludes the proof.
\hfill$\square$\\

\CYtwo{Finally, let us recall} the following theorem from Sion~\cite{Sion1958} (see also Ref.~\cite{Hidetoshi1988}):

\newpage

\begin{theorem}{\em(Sion's Minimax Theorem~\cite{Sion1958})}\label{Sion_Theorem}
Suppose $X$ is a compact convex subset of a linear topological space, and $Y$ is a convex subset of a (possibly different) linear topological space.
Let $f:X\times Y\to\mathbb{R}$ be a real-valued function such that
\begin{itemize}
\item For every $x$, the function $f(x,\cdot)$ is upper semi-continuous and quasi-concave on $Y$.
\item For every $y$, the function $f(\cdot,y)$ is lower semi-continuous and quasi-convex on $X$.
\end{itemize}
Then we have
\begin{align}
\min_{x\in X}\sup_{y\in Y}f(x,y) = \sup_{y\in Y}\min_{x\in X}f(x,y).
\end{align}
\end{theorem}

Now, we can detail the proof of 
Result~\ref{Result:conversion iff conditions}.\\

\noindent{\em Proof of 
Result~\ref{Result:conversion iff conditions}.}
First, if $\Pi({\bm\mcE}^{\AtoB{A}{B}})={\bm\mcN}^{\AtoB{A}{B}}$ for some \mbox{$\Pi\in{\bf AO}^{\AtoB{A}{B}}$}, then \eqref{Eq:iff} holds by definition. Hence, it suffices to prove the opposite direction. Suppose \eqref{Eq:iff} holds. Then we have, for every POVM ${\bf Q}^{AB}$,
\begin{align}
\max_{\Pi\in{\bf AO}^{\AtoB{A}{B}}}P_{\rm succ}[\Pi({\bm\mcE}^{\AtoB{A}{B}}),{\bm Q}^{AB}]&\ge\max_{\Pi\in{\bf AO}^{\AtoB{A}{B}}}P_{\rm succ}[\Pi({\bm\mcN}^{\AtoB{A}{B}}),{\bm Q}^{AB}]\nonumber\\
&\ge P_{\rm succ}({\bm\mcN}^{\AtoB{A}{B}},{\bm Q}^{AB}).
\end{align}
Moving everyting to the left-hand side and minimising over ${\bf Q}^{AB}$, we obtain
\begin{align}
\min_{{\bf Q}^{AB}}\left(\max_{\Pi\in{\bf AO}^{\AtoB{A}{B}}}P_{\rm succ}\left[\Pi({\bm\mcE}^{\AtoB{A}{B}}),{\bm Q}^{AB}\right]-P_{\rm succ}\left({\bm\mcN}^{\AtoB{A}{B}},{\bm Q}^{AB}\right)\right)\ge0.
\end{align}
Using the definition of $P_{\rm succ}$ [as given in \eqref{Eq: P_succ}], this can be rewritten as
\begin{align}
&\min_{{\bf Q}^{AB}}\max_{\Pi\in{\bf AO}^{\AtoB{A}{B}}}f({\bf Q}^{AB},\Pi)\coloneqq\nonumber\\
&\min_{{\bf Q}^{AB}}\max_{\Pi\in{\bf AO}^{\AtoB{A}{B}}}\sum_{a=1}^{|\mathcal{A}|}\tr\left(Q_a^{AB}\left[{\rm id}^A\otimes(\CYtwo{\Pi({\bm\mcE})_a}-\mathcal{N}_a)^{\AtoB{A}{B}}\right](\Phi_+^{AA})\right)\ge0,
\end{align}
where we define the function
\begin{align}
f({\bf Q}^{AB},\Pi)\coloneqq\sum_{a=1}^{|\mathcal{A}|}\tr\left(Q_a^{AB}\left[{\rm id}^A\otimes(\CYtwo{\Pi({\bm\mcE})_a}-\mathcal{N}_a)^{\AtoB{A}{B}}\right](\Phi_+^{AA})\right). 
\end{align}
Then one can check that $f({\bf Q}^{AB},\cdot)$ and $f(\cdot,\Pi)$ are both linear, meaning that it satisfies all the conditions required by Sion's Minimax Theorem (Thoerem~\ref{Sion_Theorem}). 
Moreover, \CYtwo{${\bf AO}^{\AtoB{A}{B}}={\rm conv}(\text{\bf pre-}{\bf AO}^{\AtoB{A}{B}})$} is convex, and the set of all strategies ${\bf Q}^{AB}$ is compact and convex. 
This means that Sion's Minimax Theorem (Thoerem~\ref{Sion_Theorem}) can be applied, resulting in 
\begin{align}\label{Eq: conversion proof 001 SM}
\max_{\Pi\in{\bf AO}^{\AtoB{A}{B}}}\min_{{\bf Q}^{AB}}f({\bf Q}^{AB},\Pi)\ge0.
\end{align}
Now, we define another function on supermaps
\begin{align}
g(\Pi)\coloneqq\min_{{\bf Q}^{AB}}f({\bf Q}^{AB},\Pi).
\end{align} One can further check that $g(\cdot)$ is a continuous function in the metric topology \CYtwo{induced by the metric $D$} defined in \eqref{Eq:metric}.
Formally, this can be seen by the following calculation for any two supermaps $\Pi,\Pi'$ [without loss of generality, we can assume $g(\Pi)\ge g(\Pi')$]:
\begin{align}
&\big|g(\Pi)-g(\Pi')\big|=g(\Pi)-g(\Pi')\nonumber\\
&=\min_{{\bf Q}^{AB}}f({\bf Q}^{AB},\Pi)-\min_{\widetilde{\bf Q}^{AB}}f(\widetilde{\bf Q}^{AB},\Pi')\nonumber\\
&=\min_{{\bf Q}^{AB}}f({\bf Q}^{AB},\Pi)+\max_{\widetilde{\bf Q}^{AB}}\left[-f(\widetilde{\bf Q}^{AB},\Pi')\right]\nonumber\\
&=\max_{\widetilde{\bf Q}^{AB}}\left[\min_{{\bf Q}^{AB}}f({\bf Q}^{AB},\Pi)-f(\widetilde{\bf Q}^{AB},\Pi')\right]\nonumber\\
&\le\max_{\widetilde{\bf Q}^{AB}}\left[f(\widetilde{\bf Q}^{AB},\Pi)-f(\widetilde{\bf Q}^{AB},\Pi')\right]\nonumber\\
&=\max_{\widetilde{\bf Q}^{AB}}\sum_{a=1}^{|\mathcal{A}|}\tr\left(\widetilde{Q}_a^{AB}\left[{\rm id}^A\otimes\CYtwo{\left(\Pi({\bm\mcE})_a-\Pi'({\bm\mcE})_a\right)^{\AtoB{A}{B}}}\right](\Phi_+^{AA})\right)\nonumber\\
&\CYtwo{\le d_Ad_B\max_{\widetilde{\bf Q}^{AB}}\sum_{a=1}^{|\mathcal{A}|}\tr\left(\frac{\widetilde{Q}_a^{AB}}{{\rm tr}(\widetilde{Q}_a^{AB})}\left[{\rm id}^A\otimes\CYtwo{\left(\Pi({\bm\mcE})_a-\Pi'({\bm\mcE})_a\right)^{\AtoB{A}{B}}}\right](\Phi_+^{AA})\right)}\nonumber\\
&\le d_Ad_B\sum_{a=1}^{|\mathcal{A}|}\norm{\left[{\rm id}^A\otimes\CYtwo{\left(\Pi({\bm\mcE})_a-\Pi'({\bm\mcE})_a\right)^{\AtoB{A}{B}}}\right](\Phi_+^{AA})}_\infty\nonumber\\
&\le d_Ad_BD(\Pi,\Pi'),
\end{align}
meaning that $g$ is (Lipschitz) continuous in the metric $D$ defined in \eqref{Eq:metric}.
Note that, in the fourth line, we use the fact that $\min_{{\bf Q}^{AB}}f({\bf Q}^{AB},\Pi)$ is a constant for the maximisation over $\widetilde{\bf Q}^{AB}$.
In the seventh line, we use the fact that ${\rm tr}(\widetilde{Q}_a^{AB})\le d_Ad_B$.

Now, in the same metric topology induced by \eqref{Eq:metric}, $g$ is a continuous function on the compact set ${\bf AO}^{\AtoB{A}{B}}$ \CYtwo{(whose compactness is from Lemma~\ref{lemma:compact})}, meaning that there is at least one $\Pi_*\in{\bf AO}^{\AtoB{A}{B}}$ achieving the maximum, namely, \mbox{$g(\Pi_*)=\max_{\Pi\in{\bf AO}^{\AtoB{A}{B}}}g(\Pi)$}. 
Hence, \eqref{Eq: conversion proof 001 SM} becomes
\begin{align}
\min_{{\bf Q}^{AB}}\sum_{a=1}^{|\mathcal{A}|}\tr\left(Q_a^{AB}\left[{\rm id}^A\otimes(\CYtwo{\Pi_*({\bm\mcE})_a}-\mathcal{N}_a)^{\AtoB{A}{B}}\right](\Phi_+^{AA})\right)\ge0.
\end{align}
One can observe that, for each fixed index $a$ and an arbitrary operator $0\preceq P^{AB}\preceq\mathds{1}^{AB}$, the set ${\bf Q}^{AB}=\{Q_m^{AB}\}_m$ with $Q_m^{AB} = \delta_{am}P^{AB}$ and $Q_{\emptyset}^{AB}=\mathds{1}^{AB}-P^{AB}$ is a valid POVM. Hence, the above inequality implies that, for every $a$ and every operator $0\preceq P^{AB}\preceq\mathds{1}^{AB}$, 
\begin{align}
\tr\left(P^{AB}\left[{\rm id}^A\otimes(\Pi_*({\bm\mcE})_a-\mathcal{N}_a)^{\AtoB{A}{B}}\right](\Phi_+^{AA})\right)\ge0.
\end{align}
The above is equivalent to
\begin{align}\label{ineq01}
\left[{\rm id}^A\otimes(\CYtwo{\Pi_*({\bm\mcE})_a}-\mathcal{N}_a)^{\AtoB{A}{B}}\right](\Phi_+^{AA})
\succeq0\quad\forall\,a.
\end{align}
In particular, this also means its trace is non-negative; namely,
\begin{align}\label{Eq: trace non-negative condition SM}
\tr\Big(\left[{\rm id}^A\otimes(\CYtwo{\Pi_*({\bm\mcE})_a}-\mathcal{N}_a)^{\AtoB{A}{B}}\right](\Phi_+^{AA})\Big)\ge0\quad\forall\,a.
\end{align}
Now, we argue that the above needs to be an {\em equality} for {\em every} $a$.
Suppose this was not the case, meaning that there is an index $a_\star$ such that the trace is strictly positive, 
$\tr\Big(\left[{\rm id}^A\otimes(\CYtwo{\Pi_*({\bm\mcE})_{a_\star}}-\mathcal{N}_{a_\star})^{\AtoB{A}{B}}\right](\Phi_+^{AA})\Big)>0$.
But then we recall that $\sum_a\tr\Big(\left[{\rm id}^A\otimes\CYtwo{\Pi_*({\bm\mcE}^{\AtoB{A}{B}})_{a}}\right](\Phi_+^{AA})\Big)=1=\sum_a\tr\Big(\left[{\rm id}^A\otimes\mathcal{N}_{a}^{\AtoB{A}{B}}\right](\Phi_+^{AA})\Big)$ since both $\Pi({\bm\mcE}^{\AtoB{A}{B}})$ and ${\bm\mcN}^{\AtoB{A}{B}}$ are instruments.
Namely, we must have $\sum_a\tr\Big(\left[{\rm id}^A\otimes(\CYtwo{\Pi_*({\bm\mcE})_a}-\mathcal{N}_a)^{\AtoB{A}{B}}\right](\Phi_+^{AA})\Big)=0$.
The existence of $a_\star$ thus implies the existence of another index $b$ such that 
$\tr\Big(\left[{\rm id}^A\otimes(\CYtwo{\Pi_*({\bm\mcE})_{b}}-\mathcal{N}_{b})^{\AtoB{A}{B}}\right](\Phi_+^{AA})\Big)<0$, leading to a contradiction with \eqref{Eq: trace non-negative condition SM}.
From here, we conclude that, for every $a$,
\begin{align}
\tr\Big(\left[{\rm id}^A\otimes\Pi_*({\bm\mcE}^{\AtoB{A}{B}})_{a}\right](\Phi_+^{AA})\Big)=\tr\Big(\left[{\rm id}^A\otimes\mathcal{N}_{a}^{\AtoB{A}{B}}\right](\Phi_+^{AA})\Big).
\end{align}
Together with \eqref{ineq01}, we further conclude that, for every $a$, the operator $\left[{\rm id}^A\otimes(\CYtwo{\Pi_*({\bm\mcE})_a}-\mathcal{N}_a)^{\AtoB{A}{B}}\right](\Phi_+^{AA})$
is trace-less and semi-definite positive---the only such operator is the zero operator. Finally, since this is a Choi operator~\cite{Choi1975,Jamiolkowski1972}, this means that $[\CYtwo{\Pi_*({\bm\mcE})_a}-\mathcal{N}_a]^{\AtoB{A}{B}}(\cdot)$ is the zero map, thereby implying that $\CYtwo{\Pi_*({\bm\mcE}^{\AtoB{A}{B}})_a(\cdot)}=\mathcal{N}_a^{\AtoB{A}{B}}(\cdot)$ for every $a$ and completing the proof.
\hfill$\square$

\bibliography{Ref.bib}

\begin{thebibliography}{38}%
\makeatletter
\providecommand \@ifxundefined [1]{%
 \@ifx{#1\undefined}
}%
\providecommand \@ifnum [1]{%
 \ifnum #1\expandafter \@firstoftwo
 \else \expandafter \@secondoftwo
 \fi
}%
\providecommand \@ifx [1]{%
 \ifx #1\expandafter \@firstoftwo
 \else \expandafter \@secondoftwo
 \fi
}%
\providecommand \natexlab [1]{#1}%
\providecommand \enquote  [1]{``#1''}%
\providecommand \bibnamefont  [1]{#1}%
\providecommand \bibfnamefont [1]{#1}%
\providecommand \citenamefont [1]{#1}%
\providecommand \href@noop [0]{\@secondoftwo}%
\providecommand \href [0]{\begingroup \@sanitize@url \@href}%
\providecommand \@href[1]{\@@startlink{#1}\@@href}%
\providecommand \@@href[1]{\endgroup#1\@@endlink}%
\providecommand \@sanitize@url [0]{\catcode `\\12\catcode `\$12\catcode `\&12\catcode `\#12\catcode `\^12\catcode `\_12\catcode `\%12\relax}%
\providecommand \@@startlink[1]{}%
\providecommand \@@endlink[0]{}%
\providecommand \url  [0]{\begingroup\@sanitize@url \@url }%
\providecommand \@url [1]{\endgroup\@href {#1}{\urlprefix }}%
\providecommand \urlprefix  [0]{URL }%
\providecommand \Eprint [0]{\href }%
\providecommand \doibase [0]{https://doi.org/}%
\providecommand \selectlanguage [0]{\@gobble}%
\providecommand \bibinfo  [0]{\@secondoftwo}%
\providecommand \bibfield  [0]{\@secondoftwo}%
\providecommand \translation [1]{[#1]}%
\providecommand \BibitemOpen [0]{}%
\providecommand \bibitemStop [0]{}%
\providecommand \bibitemNoStop [0]{.\EOS\space}%
\providecommand \EOS [0]{\spacefactor3000\relax}%
\providecommand \BibitemShut  [1]{\csname bibitem#1\endcsname}%
\let\auto@bib@innerbib\@empty
\bibitem [{\citenamefont {Chitambar}\ and\ \citenamefont {Gour}(2019)}]{ChitambarRMP2019}%
  \BibitemOpen
  \bibfield  {author} {\bibinfo {author} {\bibfnamefont {E.}~\bibnamefont {Chitambar}}\ and\ \bibinfo {author} {\bibfnamefont {G.}~\bibnamefont {Gour}},\ }\bibfield  {title} {\bibinfo {title} {Quantum resource theories},\ }\href {https://doi.org/10.1103/RevModPhys.91.025001} {\bibfield  {journal} {\bibinfo  {journal} {Rev. Mod. Phys.}\ }\textbf {\bibinfo {volume} {91}},\ \bibinfo {pages} {025001} (\bibinfo {year} {2019})}\BibitemShut {NoStop}%
\bibitem [{\citenamefont {Khandelwal}\ and\ \citenamefont {Tavakoli}(2025)}]{Khandelwal2025PRL}%
  \BibitemOpen
  \bibfield  {author} {\bibinfo {author} {\bibfnamefont {S.}~\bibnamefont {Khandelwal}}\ and\ \bibinfo {author} {\bibfnamefont {A.}~\bibnamefont {Tavakoli}},\ }\bibfield  {title} {\bibinfo {title} {Simulating quantum instruments with projective measurements and quantum postprocessing},\ }\href {https://doi.org/10.1103/bhr5-g71p} {\bibfield  {journal} {\bibinfo  {journal} {Phys. Rev. Lett.}\ }\textbf {\bibinfo {volume} {135}},\ \bibinfo {pages} {040202} (\bibinfo {year} {2025})}\BibitemShut {NoStop}%
\bibitem [{\citenamefont {G\"uhne}\ \emph {et~al.}(2023)\citenamefont {G\"uhne}, \citenamefont {Haapasalo}, \citenamefont {Kraft}, \citenamefont {Pellonp\"a\"a},\ and\ \citenamefont {Uola}}]{Otfried2021Rev}%
  \BibitemOpen
  \bibfield  {author} {\bibinfo {author} {\bibfnamefont {O.}~\bibnamefont {G\"uhne}}, \bibinfo {author} {\bibfnamefont {E.}~\bibnamefont {Haapasalo}}, \bibinfo {author} {\bibfnamefont {T.}~\bibnamefont {Kraft}}, \bibinfo {author} {\bibfnamefont {J.-P.}\ \bibnamefont {Pellonp\"a\"a}},\ and\ \bibinfo {author} {\bibfnamefont {R.}~\bibnamefont {Uola}},\ }\bibfield  {title} {\bibinfo {title} {Colloquium: Incompatible measurements in quantum information science},\ }\href {https://doi.org/10.1103/RevModPhys.95.011003} {\bibfield  {journal} {\bibinfo  {journal} {Rev. Mod. Phys.}\ }\textbf {\bibinfo {volume} {95}},\ \bibinfo {pages} {011003} (\bibinfo {year} {2023})}\BibitemShut {NoStop}%
\bibitem [{\citenamefont {Hsieh}\ \emph {et~al.}()\citenamefont {Hsieh}, \citenamefont {Uola},\ and\ \citenamefont {Skrzypczyk}}]{Hsieh-IP}%
  \BibitemOpen
  \bibfield  {author} {\bibinfo {author} {\bibfnamefont {C.-Y.}\ \bibnamefont {Hsieh}}, \bibinfo {author} {\bibfnamefont {R.}~\bibnamefont {Uola}},\ and\ \bibinfo {author} {\bibfnamefont {P.}~\bibnamefont {Skrzypczyk}},\ }\href@noop {} {\bibinfo {title} {Quantum complementarity: A novel resource for unambiguous exclusion and encryption}},\ \Eprint {https://arxiv.org/abs/2309.11968} {arXiv:2309.11968} \BibitemShut {NoStop}%
\bibitem [{\citenamefont {Skrzypczyk}\ \emph {et~al.}(2019)\citenamefont {Skrzypczyk}, \citenamefont {\ifmmode \check{S}\else \v{S}\fi{}upi\ifmmode~\acute{c}\else \'{c}\fi{}},\ and\ \citenamefont {Cavalcanti}}]{Skrzypczyk2019}%
  \BibitemOpen
  \bibfield  {author} {\bibinfo {author} {\bibfnamefont {P.}~\bibnamefont {Skrzypczyk}}, \bibinfo {author} {\bibfnamefont {I.}~\bibnamefont {\ifmmode \check{S}\else \v{S}\fi{}upi\ifmmode~\acute{c}\else \'{c}\fi{}}},\ and\ \bibinfo {author} {\bibfnamefont {D.}~\bibnamefont {Cavalcanti}},\ }\bibfield  {title} {\bibinfo {title} {All sets of incompatible measurements give an advantage in quantum state discrimination},\ }\href {https://doi.org/10.1103/PhysRevLett.122.130403} {\bibfield  {journal} {\bibinfo  {journal} {Phys. Rev. Lett.}\ }\textbf {\bibinfo {volume} {122}},\ \bibinfo {pages} {130403} (\bibinfo {year} {2019})}\BibitemShut {NoStop}%
\bibitem [{\citenamefont {Tavakoli}\ and\ \citenamefont {Uola}(2020)}]{Armin2020}%
  \BibitemOpen
  \bibfield  {author} {\bibinfo {author} {\bibfnamefont {A.}~\bibnamefont {Tavakoli}}\ and\ \bibinfo {author} {\bibfnamefont {R.}~\bibnamefont {Uola}},\ }\bibfield  {title} {\bibinfo {title} {Measurement incompatibility and steering are necessary and sufficient for operational contextuality},\ }\href {https://doi.org/10.1103/PhysRevResearch.2.013011} {\bibfield  {journal} {\bibinfo  {journal} {Phys. Rev. Res.}\ }\textbf {\bibinfo {volume} {2}},\ \bibinfo {pages} {013011} (\bibinfo {year} {2020})}\BibitemShut {NoStop}%
\bibitem [{\citenamefont {Buscemi}\ \emph {et~al.}(2020)\citenamefont {Buscemi}, \citenamefont {Chitambar},\ and\ \citenamefont {Zhou}}]{Buscemi2020PRL}%
  \BibitemOpen
  \bibfield  {author} {\bibinfo {author} {\bibfnamefont {F.}~\bibnamefont {Buscemi}}, \bibinfo {author} {\bibfnamefont {E.}~\bibnamefont {Chitambar}},\ and\ \bibinfo {author} {\bibfnamefont {W.}~\bibnamefont {Zhou}},\ }\bibfield  {title} {\bibinfo {title} {Complete resource theory of quantum incompatibility as quantum programmability},\ }\href {https://doi.org/10.1103/PhysRevLett.124.120401} {\bibfield  {journal} {\bibinfo  {journal} {Phys. Rev. Lett.}\ }\textbf {\bibinfo {volume} {124}},\ \bibinfo {pages} {120401} (\bibinfo {year} {2020})}\BibitemShut {NoStop}%
\bibitem [{\citenamefont {Ji}\ and\ \citenamefont {Chitambar}(2024)}]{Ji2024PRXQ}%
  \BibitemOpen
  \bibfield  {author} {\bibinfo {author} {\bibfnamefont {K.}~\bibnamefont {Ji}}\ and\ \bibinfo {author} {\bibfnamefont {E.}~\bibnamefont {Chitambar}},\ }\bibfield  {title} {\bibinfo {title} {Incompatibility as a resource for programmable quantum instruments},\ }\href {https://doi.org/10.1103/PRXQuantum.5.010340} {\bibfield  {journal} {\bibinfo  {journal} {PRX Quantum}\ }\textbf {\bibinfo {volume} {5}},\ \bibinfo {pages} {010340} (\bibinfo {year} {2024})}\BibitemShut {NoStop}%
\bibitem [{\citenamefont {Takagi}\ and\ \citenamefont {Regula}(2019)}]{Takagi2019}%
  \BibitemOpen
  \bibfield  {author} {\bibinfo {author} {\bibfnamefont {R.}~\bibnamefont {Takagi}}\ and\ \bibinfo {author} {\bibfnamefont {B.}~\bibnamefont {Regula}},\ }\bibfield  {title} {\bibinfo {title} {General resource theories in quantum mechanics and beyond: Operational characterization via discrimination tasks},\ }\href {https://doi.org/10.1103/PhysRevX.9.031053} {\bibfield  {journal} {\bibinfo  {journal} {Phys. Rev. X}\ }\textbf {\bibinfo {volume} {9}},\ \bibinfo {pages} {031053} (\bibinfo {year} {2019})}\BibitemShut {NoStop}%
\bibitem [{\citenamefont {Oszmaniec}\ and\ \citenamefont {Biswas}(2019)}]{Oszmaniec2019Quantum}%
  \BibitemOpen
  \bibfield  {author} {\bibinfo {author} {\bibfnamefont {M.}~\bibnamefont {Oszmaniec}}\ and\ \bibinfo {author} {\bibfnamefont {T.}~\bibnamefont {Biswas}},\ }\bibfield  {title} {\bibinfo {title} {Operational relevance of resource theories of quantum measurements},\ }\href {https://doi.org/10.22331/q-2019-04-26-133} {\bibfield  {journal} {\bibinfo  {journal} {{Quantum}}\ }\textbf {\bibinfo {volume} {3}},\ \bibinfo {pages} {133} (\bibinfo {year} {2019})}\BibitemShut {NoStop}%
\bibitem [{\citenamefont {Ducuara}\ and\ \citenamefont {Skrzypczyk}(2020)}]{Ducuara2020PRL}%
  \BibitemOpen
  \bibfield  {author} {\bibinfo {author} {\bibfnamefont {A.~F.}\ \bibnamefont {Ducuara}}\ and\ \bibinfo {author} {\bibfnamefont {P.}~\bibnamefont {Skrzypczyk}},\ }\bibfield  {title} {\bibinfo {title} {Operational interpretation of weight-based resource quantifiers in convex quantum resource theories},\ }\href {https://doi.org/10.1103/PhysRevLett.125.110401} {\bibfield  {journal} {\bibinfo  {journal} {Phys. Rev. Lett.}\ }\textbf {\bibinfo {volume} {125}},\ \bibinfo {pages} {110401} (\bibinfo {year} {2020})}\BibitemShut {NoStop}%
\bibitem [{\citenamefont {Uola}\ \emph {et~al.}(2020)\citenamefont {Uola}, \citenamefont {Bullock}, \citenamefont {Kraft}, \citenamefont {Pellonp\"a\"a},\ and\ \citenamefont {Brunner}}]{Uola2020PRL}%
  \BibitemOpen
  \bibfield  {author} {\bibinfo {author} {\bibfnamefont {R.}~\bibnamefont {Uola}}, \bibinfo {author} {\bibfnamefont {T.}~\bibnamefont {Bullock}}, \bibinfo {author} {\bibfnamefont {T.}~\bibnamefont {Kraft}}, \bibinfo {author} {\bibfnamefont {J.-P.}\ \bibnamefont {Pellonp\"a\"a}},\ and\ \bibinfo {author} {\bibfnamefont {N.}~\bibnamefont {Brunner}},\ }\bibfield  {title} {\bibinfo {title} {All quantum resources provide an advantage in exclusion tasks},\ }\href {https://doi.org/10.1103/PhysRevLett.125.110402} {\bibfield  {journal} {\bibinfo  {journal} {Phys. Rev. Lett.}\ }\textbf {\bibinfo {volume} {125}},\ \bibinfo {pages} {110402} (\bibinfo {year} {2020})}\BibitemShut {NoStop}%
\bibitem [{\citenamefont {Ducuara}\ and\ \citenamefont {Skrzypczyk}(2022)}]{Ducuara2022PRXQ}%
  \BibitemOpen
  \bibfield  {author} {\bibinfo {author} {\bibfnamefont {A.~F.}\ \bibnamefont {Ducuara}}\ and\ \bibinfo {author} {\bibfnamefont {P.}~\bibnamefont {Skrzypczyk}},\ }\bibfield  {title} {\bibinfo {title} {Characterization of quantum betting tasks in terms of arimoto mutual information},\ }\href {https://doi.org/10.1103/PRXQuantum.3.020366} {\bibfield  {journal} {\bibinfo  {journal} {PRX Quantum}\ }\textbf {\bibinfo {volume} {3}},\ \bibinfo {pages} {020366} (\bibinfo {year} {2022})}\BibitemShut {NoStop}%
\bibitem [{\citenamefont {Ku}\ \emph {et~al.}(2022)\citenamefont {Ku}, \citenamefont {Hsieh}, \citenamefont {Chen}, \citenamefont {Chen},\ and\ \citenamefont {Budroni}}]{Ku2022NC}%
  \BibitemOpen
  \bibfield  {author} {\bibinfo {author} {\bibfnamefont {H.-Y.}\ \bibnamefont {Ku}}, \bibinfo {author} {\bibfnamefont {C.-Y.}\ \bibnamefont {Hsieh}}, \bibinfo {author} {\bibfnamefont {S.-L.}\ \bibnamefont {Chen}}, \bibinfo {author} {\bibfnamefont {Y.-N.}\ \bibnamefont {Chen}},\ and\ \bibinfo {author} {\bibfnamefont {C.}~\bibnamefont {Budroni}},\ }\bibfield  {title} {\bibinfo {title} {Complete classification of steerability under local filters and its relation with measurement incompatibility},\ }\href {https://doi.org/10.1038/s41467-022-32466-y} {\bibfield  {journal} {\bibinfo  {journal} {Nat. Commun.}\ }\textbf {\bibinfo {volume} {13}},\ \bibinfo {pages} {4973} (\bibinfo {year} {2022})}\BibitemShut {NoStop}%
\bibitem [{\citenamefont {Uola}\ \emph {et~al.}(2015)\citenamefont {Uola}, \citenamefont {Budroni}, \citenamefont {G\"uhne},\ and\ \citenamefont {Pellonp\"a\"a}}]{Uola2015PRL}%
  \BibitemOpen
  \bibfield  {author} {\bibinfo {author} {\bibfnamefont {R.}~\bibnamefont {Uola}}, \bibinfo {author} {\bibfnamefont {C.}~\bibnamefont {Budroni}}, \bibinfo {author} {\bibfnamefont {O.}~\bibnamefont {G\"uhne}},\ and\ \bibinfo {author} {\bibfnamefont {J.-P.}\ \bibnamefont {Pellonp\"a\"a}},\ }\bibfield  {title} {\bibinfo {title} {One-to-one mapping between steering and joint measurability problems},\ }\href {https://doi.org/10.1103/PhysRevLett.115.230402} {\bibfield  {journal} {\bibinfo  {journal} {Phys. Rev. Lett.}\ }\textbf {\bibinfo {volume} {115}},\ \bibinfo {pages} {230402} (\bibinfo {year} {2015})}\BibitemShut {NoStop}%
\bibitem [{\citenamefont {Cavalcanti}\ and\ \citenamefont {Skrzypczyk}(2016)}]{Cavalcanti2016}%
  \BibitemOpen
  \bibfield  {author} {\bibinfo {author} {\bibfnamefont {D.}~\bibnamefont {Cavalcanti}}\ and\ \bibinfo {author} {\bibfnamefont {P.}~\bibnamefont {Skrzypczyk}},\ }\bibfield  {title} {\bibinfo {title} {Quantum steering: {A} review with focus on semidefinite programming},\ }\href {https://doi.org/10.1088/1361-6633/80/2/024001} {\bibfield  {journal} {\bibinfo  {journal} {Rep. Prog. Phys.}\ }\textbf {\bibinfo {volume} {80}},\ \bibinfo {pages} {024001} (\bibinfo {year} {2016})}\BibitemShut {NoStop}%
\bibitem [{\citenamefont {Skrzypczyk}\ and\ \citenamefont {Linden}(2019)}]{Skrzypczyk2019PRL}%
  \BibitemOpen
  \bibfield  {author} {\bibinfo {author} {\bibfnamefont {P.}~\bibnamefont {Skrzypczyk}}\ and\ \bibinfo {author} {\bibfnamefont {N.}~\bibnamefont {Linden}},\ }\bibfield  {title} {\bibinfo {title} {Robustness of measurement, discrimination games, and accessible information},\ }\href {https://doi.org/10.1103/PhysRevLett.122.140403} {\bibfield  {journal} {\bibinfo  {journal} {Phys. Rev. Lett.}\ }\textbf {\bibinfo {volume} {122}},\ \bibinfo {pages} {140403} (\bibinfo {year} {2019})}\BibitemShut {NoStop}%
\bibitem [{\citenamefont {Hsieh}\ and\ \citenamefont {Chen}(2024)}]{Hsieh2024PRL}%
  \BibitemOpen
  \bibfield  {author} {\bibinfo {author} {\bibfnamefont {C.-Y.}\ \bibnamefont {Hsieh}}\ and\ \bibinfo {author} {\bibfnamefont {S.-L.}\ \bibnamefont {Chen}},\ }\bibfield  {title} {\bibinfo {title} {Thermodynamic approach to quantifying incompatible instruments},\ }\href {https://doi.org/10.1103/PhysRevLett.133.170401} {\bibfield  {journal} {\bibinfo  {journal} {Phys. Rev. Lett.}\ }\textbf {\bibinfo {volume} {133}},\ \bibinfo {pages} {170401} (\bibinfo {year} {2024})}\BibitemShut {NoStop}%
\bibitem [{\citenamefont {Ghai}\ and\ \citenamefont {Mitra}(2025)}]{Ghai2025}%
  \BibitemOpen
  \bibfield  {author} {\bibinfo {author} {\bibfnamefont {J.}~\bibnamefont {Ghai}}\ and\ \bibinfo {author} {\bibfnamefont {A.}~\bibnamefont {Mitra}},\ }\href@noop {} {\bibinfo {title} {Instrument-based quantum resources: quantification, hierarchies and towards constructing resource theories}} (\bibinfo {year} {2025}),\ \Eprint {https://arxiv.org/abs/2508.09134} {arXiv:2508.09134 [quant-ph]} \BibitemShut {NoStop}%
\bibitem [{\citenamefont {Buscemi}\ \emph {et~al.}(2023)\citenamefont {Buscemi}, \citenamefont {Kobayashi}, \citenamefont {Minagawa}, \citenamefont {Perinotti},\ and\ \citenamefont {Tosini}}]{Buscemi2023}%
  \BibitemOpen
  \bibfield  {author} {\bibinfo {author} {\bibfnamefont {F.}~\bibnamefont {Buscemi}}, \bibinfo {author} {\bibfnamefont {K.}~\bibnamefont {Kobayashi}}, \bibinfo {author} {\bibfnamefont {S.}~\bibnamefont {Minagawa}}, \bibinfo {author} {\bibfnamefont {P.}~\bibnamefont {Perinotti}},\ and\ \bibinfo {author} {\bibfnamefont {A.}~\bibnamefont {Tosini}},\ }\bibfield  {title} {\bibinfo {title} {Unifying different notions of quantum incompatibility into a strict hierarchy of resource theories of communication},\ }\href {https://doi.org/10.22331/q-2023-06-07-1035} {\bibfield  {journal} {\bibinfo  {journal} {{Quantum}}\ }\textbf {\bibinfo {volume} {7}},\ \bibinfo {pages} {1035} (\bibinfo {year} {2023})}\BibitemShut {NoStop}%
\bibitem [{\citenamefont {Lepp{\"{a}}j{\"{a}}rvi}\ and\ \citenamefont {Sedl{\'{a}}k}(2024)}]{Leppajarvi2024Quantum}%
  \BibitemOpen
  \bibfield  {author} {\bibinfo {author} {\bibfnamefont {L.}~\bibnamefont {Lepp{\"{a}}j{\"{a}}rvi}}\ and\ \bibinfo {author} {\bibfnamefont {M.}~\bibnamefont {Sedl{\'{a}}k}},\ }\bibfield  {title} {\bibinfo {title} {Incompatibility of quantum instruments},\ }\href {https://doi.org/10.22331/q-2024-02-12-1246} {\bibfield  {journal} {\bibinfo  {journal} {{Quantum}}\ }\textbf {\bibinfo {volume} {8}},\ \bibinfo {pages} {1246} (\bibinfo {year} {2024})}\BibitemShut {NoStop}%
\bibitem [{\citenamefont {Kuroiwa}\ \emph {et~al.}(2024)\citenamefont {Kuroiwa}, \citenamefont {Takagi}, \citenamefont {Adesso},\ and\ \citenamefont {Yamasaki}}]{Kuroiwa2024PRA}%
  \BibitemOpen
  \bibfield  {author} {\bibinfo {author} {\bibfnamefont {K.}~\bibnamefont {Kuroiwa}}, \bibinfo {author} {\bibfnamefont {R.}~\bibnamefont {Takagi}}, \bibinfo {author} {\bibfnamefont {G.}~\bibnamefont {Adesso}},\ and\ \bibinfo {author} {\bibfnamefont {H.}~\bibnamefont {Yamasaki}},\ }\bibfield  {title} {\bibinfo {title} {Robustness- and weight-based resource measures without convexity restriction: Multicopy witness and operational advantage in static and dynamical quantum resource theories},\ }\href {https://doi.org/10.1103/PhysRevA.109.042403} {\bibfield  {journal} {\bibinfo  {journal} {Phys. Rev. A}\ }\textbf {\bibinfo {volume} {109}},\ \bibinfo {pages} {042403} (\bibinfo {year} {2024})}\BibitemShut {NoStop}%
\bibitem [{\citenamefont {Boyd}\ and\ \citenamefont {Vandenberghe}(2004)}]{Boyd-book}%
  \BibitemOpen
  \bibfield  {author} {\bibinfo {author} {\bibfnamefont {S.}~\bibnamefont {Boyd}}\ and\ \bibinfo {author} {\bibfnamefont {L.}~\bibnamefont {Vandenberghe}},\ }\href@noop {} {\emph {\bibinfo {title} {Convex Optimization}}}\ (\bibinfo  {publisher} {Cambridge University Press},\ \bibinfo {year} {2004})\BibitemShut {NoStop}%
\bibitem [{\citenamefont {Skrzypczyk}\ and\ \citenamefont {Cavalcanti}(2023)}]{SDP-textbook}%
  \BibitemOpen
  \bibfield  {author} {\bibinfo {author} {\bibfnamefont {P.}~\bibnamefont {Skrzypczyk}}\ and\ \bibinfo {author} {\bibfnamefont {D.}~\bibnamefont {Cavalcanti}},\ }\href {https://doi.org/10.1088/978-0-7503-3343-6} {\emph {\bibinfo {title} {Semidefinite Programming in Quantum Information Science}}},\ 2053-2563\ (\bibinfo  {publisher} {IOP Publishing, Bristol},\ \bibinfo {year} {2023})\BibitemShut {NoStop}%
\bibitem [{\citenamefont {Takagi}\ \emph {et~al.}(2020)\citenamefont {Takagi}, \citenamefont {Wang},\ and\ \citenamefont {Hayashi}}]{Takagi2020PRL}%
  \BibitemOpen
  \bibfield  {author} {\bibinfo {author} {\bibfnamefont {R.}~\bibnamefont {Takagi}}, \bibinfo {author} {\bibfnamefont {K.}~\bibnamefont {Wang}},\ and\ \bibinfo {author} {\bibfnamefont {M.}~\bibnamefont {Hayashi}},\ }\bibfield  {title} {\bibinfo {title} {Application of the resource theory of channels to communication scenarios},\ }\href {https://doi.org/10.1103/PhysRevLett.124.120502} {\bibfield  {journal} {\bibinfo  {journal} {Phys. Rev. Lett.}\ }\textbf {\bibinfo {volume} {124}},\ \bibinfo {pages} {120502} (\bibinfo {year} {2020})}\BibitemShut {NoStop}%
\bibitem [{\citenamefont {Horodecki}\ \emph {et~al.}(1999)\citenamefont {Horodecki}, \citenamefont {Horodecki},\ and\ \citenamefont {Horodecki}}]{Horodecki1999}%
  \BibitemOpen
  \bibfield  {author} {\bibinfo {author} {\bibfnamefont {M.}~\bibnamefont {Horodecki}}, \bibinfo {author} {\bibfnamefont {P.}~\bibnamefont {Horodecki}},\ and\ \bibinfo {author} {\bibfnamefont {R.}~\bibnamefont {Horodecki}},\ }\bibfield  {title} {\bibinfo {title} {General teleportation channel, singlet fraction, and quasidistillation},\ }\href {https://doi.org/10.1103/PhysRevA.60.1888} {\bibfield  {journal} {\bibinfo  {journal} {Phys. Rev. A}\ }\textbf {\bibinfo {volume} {60}},\ \bibinfo {pages} {1888} (\bibinfo {year} {1999})}\BibitemShut {NoStop}%
\bibitem [{\citenamefont {Konig}\ \emph {et~al.}(2009)\citenamefont {Konig}, \citenamefont {Renner},\ and\ \citenamefont {Schaffner}}]{Konig2009ITIT}%
  \BibitemOpen
  \bibfield  {author} {\bibinfo {author} {\bibfnamefont {R.}~\bibnamefont {Konig}}, \bibinfo {author} {\bibfnamefont {R.}~\bibnamefont {Renner}},\ and\ \bibinfo {author} {\bibfnamefont {C.}~\bibnamefont {Schaffner}},\ }\bibfield  {title} {\bibinfo {title} {The operational meaning of min- and max-entropy},\ }\href {https://doi.org/10.1109/TIT.2009.2025545} {\bibfield  {journal} {\bibinfo  {journal} {IEEE Trans. Inf. Theory}\ }\textbf {\bibinfo {volume} {55}},\ \bibinfo {pages} {4337} (\bibinfo {year} {2009})}\BibitemShut {NoStop}%
\bibitem [{\citenamefont {Datta}(2009)}]{Datta2009ITIT}%
  \BibitemOpen
  \bibfield  {author} {\bibinfo {author} {\bibfnamefont {N.}~\bibnamefont {Datta}},\ }\bibfield  {title} {\bibinfo {title} {Min- and max-relative entropies and a new entanglement monotone},\ }\href {https://doi.org/10.1109/TIT.2009.2018325} {\bibfield  {journal} {\bibinfo  {journal} {IEEE Trans. Inf. Theory}\ }\textbf {\bibinfo {volume} {55}},\ \bibinfo {pages} {2816} (\bibinfo {year} {2009})}\BibitemShut {NoStop}%
\bibitem [{\citenamefont {Nielsen}(2002)}]{Nielsen2002PLA}%
  \BibitemOpen
  \bibfield  {author} {\bibinfo {author} {\bibfnamefont {M.~A.}\ \bibnamefont {Nielsen}},\ }\bibfield  {title} {\bibinfo {title} {A simple formula for the average gate fidelity of a quantum dynamical operation},\ }\href {https://doi.org/https://doi.org/10.1016/S0375-9601(02)01272-0} {\bibfield  {journal} {\bibinfo  {journal} {Phys. Lett. A}\ }\textbf {\bibinfo {volume} {303}},\ \bibinfo {pages} {249} (\bibinfo {year} {2002})}\BibitemShut {NoStop}%
\bibitem [{\citenamefont {Buscemi}\ and\ \citenamefont {Datta}(2010)}]{Buscemi2010ITIT}%
  \BibitemOpen
  \bibfield  {author} {\bibinfo {author} {\bibfnamefont {F.}~\bibnamefont {Buscemi}}\ and\ \bibinfo {author} {\bibfnamefont {N.}~\bibnamefont {Datta}},\ }\bibfield  {title} {\bibinfo {title} {The quantum capacity of channels with arbitrarily correlated noise},\ }\href {https://doi.org/10.1109/TIT.2009.2039166} {\bibfield  {journal} {\bibinfo  {journal} {IEEE Trans. Inf. Theory}\ }\textbf {\bibinfo {volume} {56}},\ \bibinfo {pages} {1447} (\bibinfo {year} {2010})}\BibitemShut {NoStop}%
\bibitem [{\citenamefont {Choi}(1975)}]{Choi1975}%
  \BibitemOpen
  \bibfield  {author} {\bibinfo {author} {\bibfnamefont {M.-D.}\ \bibnamefont {Choi}},\ }\bibfield  {title} {\bibinfo {title} {Completely positive linear maps on complex matrices},\ }\href {https://doi.org/https://doi.org/10.1016/0024-3795(75)90075-0} {\bibfield  {journal} {\bibinfo  {journal} {Linear Algebr. Appl.}\ }\textbf {\bibinfo {volume} {10}},\ \bibinfo {pages} {285} (\bibinfo {year} {1975})}\BibitemShut {NoStop}%
\bibitem [{\citenamefont {Jamio{\l}kowski}(1972)}]{Jamiolkowski1972}%
  \BibitemOpen
  \bibfield  {author} {\bibinfo {author} {\bibfnamefont {A.}~\bibnamefont {Jamio{\l}kowski}},\ }\bibfield  {title} {\bibinfo {title} {Linear transformations which preserve trace and positive semidefiniteness of operators},\ }\href {https://doi.org/https://doi.org/10.1016/0034-4877(72)90011-0} {\bibfield  {journal} {\bibinfo  {journal} {Rep. Math. Phys.}\ }\textbf {\bibinfo {volume} {3}},\ \bibinfo {pages} {275} (\bibinfo {year} {1972})}\BibitemShut {NoStop}%
\bibitem [{\citenamefont {Slater}(2014)}]{Slater2014}%
  \BibitemOpen
  \bibfield  {author} {\bibinfo {author} {\bibfnamefont {M.}~\bibnamefont {Slater}},\ }\bibinfo {title} {Lagrange multipliers revisited},\ in\ \href {https://doi.org/10.1007/978-3-0348-0439-4_14} {\emph {\bibinfo {booktitle} {Traces and Emergence of Nonlinear Programming}}},\ \bibinfo {editor} {edited by\ \bibinfo {editor} {\bibfnamefont {G.}~\bibnamefont {Giorgi}}\ and\ \bibinfo {editor} {\bibfnamefont {T.~H.}\ \bibnamefont {Kjeldsen}}}\ (\bibinfo  {publisher} {Springer Basel},\ \bibinfo {address} {Basel},\ \bibinfo {year} {2014})\ pp.\ \bibinfo {pages} {293--306}\BibitemShut {NoStop}%
\bibitem [{\citenamefont {Watrous}(2018)}]{Watrous_2018}%
  \BibitemOpen
  \bibfield  {author} {\bibinfo {author} {\bibfnamefont {J.}~\bibnamefont {Watrous}},\ }\href {https://doi.org/10.1017/9781316848142} {\emph {\bibinfo {title} {The Theory of Quantum Information}}}\ (\bibinfo  {publisher} {Cambridge University Press, Cambridge, England},\ \bibinfo {year} {2018})\BibitemShut {NoStop}%
\bibitem [{\citenamefont {Tavakoli}\ \emph {et~al.}(2024)\citenamefont {Tavakoli}, \citenamefont {Pozas-Kerstjens}, \citenamefont {Brown},\ and\ \citenamefont {Ara\'ujo}}]{Tavakoli2024RMP}%
  \BibitemOpen
  \bibfield  {author} {\bibinfo {author} {\bibfnamefont {A.}~\bibnamefont {Tavakoli}}, \bibinfo {author} {\bibfnamefont {A.}~\bibnamefont {Pozas-Kerstjens}}, \bibinfo {author} {\bibfnamefont {P.}~\bibnamefont {Brown}},\ and\ \bibinfo {author} {\bibfnamefont {M.}~\bibnamefont {Ara\'ujo}},\ }\bibfield  {title} {\bibinfo {title} {Semidefinite programming relaxations for quantum correlations},\ }\href {https://doi.org/10.1103/RevModPhys.96.045006} {\bibfield  {journal} {\bibinfo  {journal} {Rev. Mod. Phys.}\ }\textbf {\bibinfo {volume} {96}},\ \bibinfo {pages} {045006} (\bibinfo {year} {2024})}\BibitemShut {NoStop}%
\bibitem [{\citenamefont {Harrow}(2013)}]{Harrow2013}%
  \BibitemOpen
  \bibfield  {author} {\bibinfo {author} {\bibfnamefont {A.~W.}\ \bibnamefont {Harrow}},\ }\href@noop {} {\bibinfo {title} {The church of the symmetric subspace}} (\bibinfo {year} {2013}),\ \Eprint {https://arxiv.org/abs/1308.6595} {arXiv:1308.6595 [quant-ph]} \BibitemShut {NoStop}%
\bibitem [{\citenamefont {Sion}(1958)}]{Sion1958}%
  \BibitemOpen
  \bibfield  {author} {\bibinfo {author} {\bibfnamefont {M.}~\bibnamefont {Sion}},\ }\bibfield  {title} {\bibinfo {title} {On general minimax theorems},\ }\href@noop {} {\bibfield  {journal} {\bibinfo  {journal} {Pac. J. Math.}\ }\textbf {\bibinfo {volume} {8}},\ \bibinfo {pages} {171} (\bibinfo {year} {1958})}\BibitemShut {NoStop}%
\bibitem [{\citenamefont {Komiya}(1988)}]{Hidetoshi1988}%
  \BibitemOpen
  \bibfield  {author} {\bibinfo {author} {\bibfnamefont {H.}~\bibnamefont {Komiya}},\ }\bibfield  {title} {\bibinfo {title} {{Elementary proof for Sion's minimax theorem}},\ }\href {https://doi.org/10.2996/kmj/1138038812} {\bibfield  {journal} {\bibinfo  {journal} {Kodai Math. J.}\ }\textbf {\bibinfo {volume} {11}},\ \bibinfo {pages} {5} (\bibinfo {year} {1988})}\BibitemShut {NoStop}%
\end{thebibliography}%

\end{document}